\documentclass[11pt, letter]{article}
\usepackage{geometry}
\usepackage{paralist}
\usepackage{fullpage}
\usepackage{amsmath}
\usepackage{amssymb}
\usepackage{hyperref}
\usepackage{amsthm}
\usepackage{thm-restate}
\usepackage[nameinlink, capitalize]{cleveref}
\usepackage[margin=1.5cm]{caption}
\usepackage{subcaption}

\newcommand{\poly}{\operatorname{poly}}

\usepackage{tikz}
\usetikzlibrary{arrows.meta,fit,positioning}

\newtheorem{theorem}{Theorem}[section]
\newtheorem{lemma}[theorem]{Lemma}

\theoremstyle{definition}
\newtheorem{definition}[theorem]{Definition}

\title{Exponential Lower Bounds for Integer-Weighted Shortest-Paths Preservers of DAGs}
 \author{Michael Yi Wang\thanks{\texttt{ywangm@umich.edu}}\\University of Michigan \and Nicole Wein\thanks{\texttt{nswein@umich.edu. Supported by NSF CAREER award 2541910. }}\\University of Michigan }
\date{}

\begin{document}

\maketitle

\begin{abstract}

We study a \emph{graph simplification} problem introduced by Bernstein, Bodwin, and Wein [ITCS'24]. 
We start with a graph with arbitrarily large positive edge weights and the goal is to \emph{reweight} the edges to small \emph{aspect ratio} (ratio between largest and smallest weight) while preserving the \emph{shortest paths structure} (the sequence of vertices and edges along shortest paths). 

They studied whether \emph{polynomial} aspect ratio is always possible. They proved that for general graphs, both directed and undirected, it is not: there exist graphs for which any shortest-paths preserving reweighting requires \emph{exponential} aspect ratio. In contrast, they showed that every DAG (directed acyclic graph) admits a reweighting with \emph{linear} aspect ratio. However, the resulting edge weights are not integers. This  motivated them to pose the open question of whether all DAGs admit a reweighting with polynomially-bounded \emph{integer} edge weights. 

Our main result is to answer this question in the negative: we prove that there exist DAGs for which any shortest-paths preserving integer reweighting requires weights of size $2^{\Omega(n)}$. In fact, this is even true when the DAG has very simple structure: 3 layers of vertices with only 3 vertices in the middle layer. In contrast, we show that if the number of vertices in the middle layer is decreased to 2, then a linear upper bound is possible. 

We extend our exponential lower bound to the \emph{approximate} version of the problem where only a single $\alpha$-approximate shortest path in the original graph must be preserved as an exact shortest path in the reweighted graph. Our exponential lower bound holds even for \emph{any} finite approximation ratio $\alpha>1$.

\end{abstract}

\pagenumbering{gobble}
\clearpage

\pagebreak
\pagenumbering{arabic}

\section{Introduction}

In modern graph algorithms, a popular strategy is to simplify graphs during preprocessing before
proceeding to the main part of the algorithm. We study a graph simplification problem introduced by Bernstein, Bodwin, and Wein \cite{bernstein2025graphsshortestpathstructure}, where we start with a graph with arbitrarily large positive edge weights and the goal is to \emph{reweight} its edges to minimize the \emph{aspect ratio} of the edge weights while preserving
the \emph{shortest path structure}. 

The \emph{aspect ratio} of a (positively weighted) graph is the the multiplicative spread among the edge weights:

\begin{definition}[Aspect ratio] 
Let $G = (V, E, w)$ be a graph with positive edge weights. The \emph{aspect ratio} of $G$ is the quantity $\frac{\max_{e \in E} w(e)}{\min_{e \in E} w(e)}$.
\end{definition}

The notion of preserving the shortest paths structure is as follows: 

\begin{definition}[Shortest-paths preserver]
Given a graph $G = (V, E, w)$, a positively reweighted graph $H = (V, E, w_H)$ on the same vertex and edge set is a \textit{shortest-paths preserver} if, for every shortest path $\pi$ in $H$, the sequence of nodes and edges along $\pi$ is also a shortest path in $G$.
\end{definition}

Note that according to the definition, in the case of ties, only one shortest path between each pair of endpoints must be preserved.  

See \Cref{fig:shppres}, copied from~\cite{bernstein2025graphsshortestpathstructure},  for an example of a shortest-paths preserver. Importantly, after reweighting, the shortest paths can have vastly different lengths than they did originally, but the sequence of vertices and edges along these paths must be preserved.

\begin{figure}[h]
\centering
\begin{tikzpicture}[scale=.7]
\draw [fill=black] (0, 0) circle [radius=0.15];
\draw [fill=black] (2, -2) circle [radius=0.15];
\draw [fill=black] (2, 2) circle [radius=0.15];
\draw [fill=black] (4, 0) circle [radius=0.15];
\draw [ultra thick] (0, 0) -- (2, -2) -- (4, 0) -- (2, 2) -- (0, 0) -- (4, 0);

\node at (2, 0.3) {$500$};
\node at (0.7, 1.3) {$97$};
\node at (0.7, -1.3) {$53$};
\node at (3.3, 1.3) {$5$};
\node at (3.3, -1.3) {$83$};

\node at (2, -3) {\Huge $G$};

\draw [ultra thick, ->] (5, 0) -- (7, 0);

\begin{scope}[shift={(8, 0)}]
\draw [fill=black] (0, 0) circle [radius=0.15];
\draw [fill=black] (2, -2) circle [radius=0.15];
\draw [fill=black] (2, 2) circle [radius=0.15];
\draw [fill=black] (4, 0) circle [radius=0.15];
\draw [ultra thick] (0, 0) -- (2, -2) -- (4, 0) -- (2, 2) -- (0, 0) -- (4, 0);

\node at (2, -3) {\Huge $H$};

\node at (2, 0.3) {$4$};
\node at (0.7, 1.3) {$1$};
\node at (0.7, -1.3) {$2$};
\node at (3.3, 1.3) {$1$};
\node at (3.3, -1.3) {$1$};
\end{scope}
\end{tikzpicture}
\caption{\cite{bernstein2025graphsshortestpathstructure} $H$ is a shortest-paths preserver of $G$ because its shortest paths have the same vertices and edges as those in $G$. Additionally, $H$ has aspect ratio $4$ while $G$ which has aspect ratio $100$.}\label{fig:shppres} 
\end{figure}
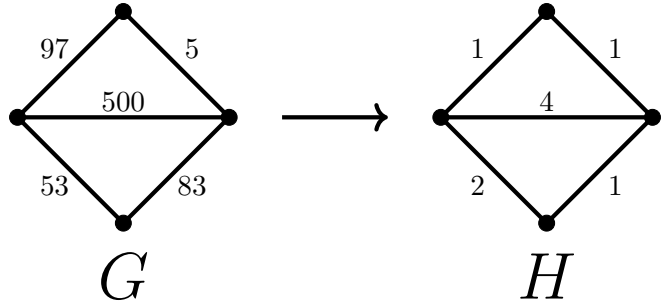

Like \cite{bernstein2025graphsshortestpathstructure}, we are interested in low aspect ratio shortest-paths preservers from an existential perspective: can we prove that every graph from a given class has a shortest-paths preserver with certain aspect ratio? In particular, the main question is whether or not \emph{polynomial} aspect ratio is possible. 

\subsection{Motivation}

Many algorithms for problems involving shortest paths incur a dependence on the aspect ratio, or only work when aspect ratio is polynomial. Common reasons for this include the use of bucketing techniques to group shortest paths by length, or matrix-multiplication-based techniques that rely on bounded matrix entries.
Some specific examples noted by~\cite{bernstein2025graphsshortestpathstructure} of shortest-path-related problems that have incurred a dependence on aspect ratio for these reasons
include $(1+\varepsilon)$-hopsets \cite{BW23, KP22a}, roundtrip spanners \cite{RTZ08, ZL18, CDG20}, the All-Pairs Shortest Paths (APSP) problem \cite{Zwick02, SZ99}, dynamic shortest paths (e.g. \cite{HenzingerKN14,BernsteinGW20,ChuzhoyZ23}), and distributed shorted paths (e.g. \cite{ForsterN18,CaoFR21}).

This motivated Bernstein, Bodwin, and Wein~\cite{bernstein2025graphsshortestpathstructure} to define low aspect ratio shortest-paths preservers. If such objects existed, then they could yield a general-purpose reduction for removing the dependence on aspect ratio from many known results. On the other hand, if such objects did not exist, this would limit the shortest paths structures realizable by low aspect ratio graphs, which would be an interesting graph-theoretic phenomenon, and also potentially help with the design of algorithms specifically for low aspect ratio graphs.

More broadly, this research agenda falls into the realm of \emph{graph simplification}: representing a graph by another simpler graph that preserves an important property of the original graph. Some other examples of such objects include 
spanners \cite{Pu89jacm,PU89sicomp}, distance preservers \cite{CE06,AB18}, flow/cut/spectral sparsifiers \cite{ST11,BK96}, mimicking networks \cite{CDKLLPSV21, HKNR98}, terminal minor sparsifiers \cite{KNZ14}, hopsets and shortcut-sets \cite{Thorup92,UY91}, expander decompositions \cite{KVV04,ST11}, tree covers/embeddings \cite{karp19892k,MR1313480}, and DAG covers \cite{MR4928607}.

See the introduction of~\cite{bernstein2025graphsshortestpathstructure} for further discussion of shortest-paths preservers and the surrounding literature. 

\subsection{Known Results}
Bernstein, Bodwin, and Wein~\cite{bernstein2025graphsshortestpathstructure} proved that for general graphs, both directed and undirected, it is not possible to achieve $\mathrm{poly}(n)$ aspect ratio. In particular, there exist $n$-node directed and undirected graphs such that any shortest-paths preserver has  aspect ratio $2^{\Omega(n)}$.

However, for DAGs (directed acyclic graphs), the story is starkly different: \cite{bernstein2025graphsshortestpathstructure} showed a linear upper bound; that is, every $n$-node DAG has a shortest-paths preserver with aspect ratio $O(n)$. Despite this strong upper bound, which also comes with a linear time algorithm, the authors did not identify any algorithmic applications. 

The main barrier against getting algorithmic applications from their method is that even though the aspect ratio is linear, the edge weights are huge. In particular, if $W$ is the original maximum edge weight, each of the original edge weights is \emph{increased} by a quantity between $W$ and $Wn$. Of course, one can always scale down edge weights to any desired range without disrupting the shortest path structure, but scaling them down to polynomial renders them non-integers. This is important because for many shortest-paths related problems, integer edge weights are easier to handle than general real weights because integers enable certain techniques such as bucketing (e.g.~\cite{MR1815738} and many more). For this reason, obtaining low aspect ratio shortest-paths preservers with polynomially-bounded \emph{integer} weights for DAGs would be a more useful tool than the known non-integer counterpart.

In fact, although \cite{bernstein2025graphsshortestpathstructure} did not identify an application for their non-integer linear aspect ratio result for DAGs, they did identify a concrete example of an application if an integer result were to be proven, assuming it could be computed efficiently:
extending the parallel/distributed results in Corollary 1.8 and 1.9 of \cite{RozhonHMGZ23} to work for arbitrary weighted graphs.

Motivated by such potential applications and the shortcomings of the non-integer result, \cite{bernstein2025graphsshortestpathstructure} explicitly posed the following open problem:

\begin{center}
\begin{minipage}{0.9\textwidth}
\begin{center}
    \emph{Does every $n$-node DAG have a shortest-paths preserver $H$ with \textbf{integer} edge weights in the range $[1, \dots, \poly(n)]$? If not, does this hold if we allow stretch $\alpha$?}
    \end{center}
    \end{minipage}
    \end{center}

Their question about stretch refers to the \emph{approximate} version of shortest-paths preservers, where the requirement is that at least one of the $\alpha$-approximate shortest paths between each pair of endpoints in the original graph $G$ is preserved as an exact shortest path in $H$:

\begin{definition}[$\alpha$-stretch shortest-paths preserver] 
Given a graph $G = (V,E,w)$, we say $H = (V,E,w_H)$ is an \emph{$\alpha$-stretch shortest-paths preserver} of $G$ if every shortest path in $H$ is also an $\alpha$-approximate shortest path in $G$. That is, if $\pi$ is a shortest path from $s$ to $t$ in $H$, then $w_G(\pi) \le \alpha \cdot \mathrm{dist}_G(s,t)$
\end{definition}

In \cite{bernstein2025graphsshortestpathstructure} they showed that for general directed graphs, there is an exponential lower bound of aspect ratio $2^{\Omega(n)}$ for the $\alpha$-stretch version even for any finite $\alpha$. For undirected graphs, they provided such an exponential lower bound only for $\alpha\leq 13/12$. They did not explicitly prove any results about the $\alpha$-stretch version for DAGs because they already achieved a linear upper bound for the exact version, and the $\alpha$-stretch version is strictly easier. For DAGs with integer edge weights, they also did not prove any results for the $\alpha$-stretch version, as that is their above open question. In summary, no nontrivial bounds were known for integer-weighted shortest-paths preservers for DAGs, either in the $\alpha$-stretch or exact version.

\subsection{Our Results}

Our main result is to answer the above two open questions in the negative. We prove that there exist DAGs such that any shortest-paths preserver with integer edge weights must have exponentially large maximum edge weight. We prove that this is still the case even for $\alpha$-stretch shortest-paths preservers for \emph{any} finite $\alpha>1$. The new landscape of results is striking: restricting from real edge weights to integer edge weights yields a sudden transition from linear to exponential.

Our theorem statements are actually stronger than the above informal statement in two ways. Firstly, instead of requiring $H$ to have integer edge weights, we allow $H$ to have real edge weights and only require that every shortest path and not-shortest path with the same pair of endpoints differ by at least 1. Formally, such shortest-paths preservers are defined below.

\begin{definition}[Gap-1 shortest-paths preservers]
Given a graph $G = (V,E,w)$, a positively reweighted graph $H = (V,E,w_H)$ on the same vertex and edge set is a \emph{gap-1 shortest-paths preserver} if $H$ satisfies the following:
\begin{itemize}
    \item 
    $H$ is a shortest-paths preserver.

    \item 
    For any shortest path $\pi$ in $H$, and any not-shortest path $\pi'$ in $H$ with the same endpoints as $\pi$, we have $w_H(\pi) + 1 \le w_H(\pi')$.
    
\end{itemize}
\end{definition}

Note that if a shortest-paths preserver $H$ has integer weights, then it is automatically a gap-1 shortest-paths preserver, so a lower bound for gap-1 shortest paths preservers would extend to a lower bound for integer-weighted shortest paths preservers. 

Secondly, and more importantly, we prove our exponential lower bounds (for both the exact and $\alpha$-stretch versions) even for an extremely simple class of DAGs: 3-layered DAGs with only 3 vertices in the middle layer. That is, the vertex set is of the form $V = V_1 \cup V_2 \cup V_3$, edges only go from $V_1$ to $V_2$ and from $V_2$ to $V_3$, and $|V_2| = 3$.

Our exponential lower bound for the exact version is stated as follows:

\begin{restatable}[Exponential lower bound for 3-layered DAGs ]{theorem}{maintheorem}\label{thm:main}
    There are 3-layered DAGs $G = (V = V_1 \cup V_2 \cup V_3, E, w)$ such that $|V_2| = 3$ and any gap-1 shortest-paths preserver $H = (V, E, w_H)$ 
    satisfies 
    \[
    \max_{e \in E} w_H(e) \ge 2^{\Omega(|V|)}.
    \]
\end{restatable}

Our exponential lower bound for the $\alpha$-stretch version is stated as follows:

\begin{restatable}[Exponential lower bound for $\alpha$-stretch version for 3-layered DAGs]{theorem}{approx}\label{thm:approx}

For any $\alpha > 1$, there are 3-layered DAGs $G = (V = V_1 \cup V_2 \cup V_3, E, w)$ such that $|V_2| = 3$ and any $\alpha$-stretch gap-1 shortest-paths preserver $H = (V, E, w_H)$ satisfies 
    \[
    \max_{e \in E} w_H(e) \ge 2^{\Omega(|V|)}.
    \]
\end{restatable}

Additionally, we prove that the class of graphs in our lower bounds (3-layered DAGs with 3 vertices in the middle layer) is precisely on the boundary between exponential and polynomial: If we remove even one vertex from the middle layer, we show that a \emph{linear} upper bound is possible. 

\begin{restatable}[Linear upper bound for 3-layered DAGs with integer weights when $|V_2| = 2$]{theorem}{upperbound}\label{thm:upper_bound}
Every 3-layered DAG $G = (V = V_1 \cup V_2 \cup V_3, E,w)$ with $|V_2| = 2$ has a shortest-paths preserver $H = (V,E,w_H)$ with integer edge weights in the range $[1,\dots, |V_1| + |V_3| + 1]$. 
    
\end{restatable}

\subsection{Our Techniques}
\label{subsec:techniques}
On first glance, one might expect that every DAG has a shortest-paths preserver with linearly bounded integer edge weights, via rounding the existing non-integer linear upper bound. It sounds reasonable that there would be some slack available for rounding because we only need to preserve the identity of a shortest path between each pair of endpoints, not the relative lengths of any of the other paths between those endpoints. Furthermore, even without a linear (or polynomial) upper bound for general DAGs, one might at least expect a polynomial upper bound for the simple class of 3-layered DAGs with only 3 vertices in the middle layer. Indeed, we found it surprising that there are in fact exponential lower bounds.

As a starting point for proving an exponential lower bound for DAGs with integer edge weights, it is natural to consider the known exponential lower bounds for general (directed and undirected) graphs. These lower bounds are quite simple and we briefly outline them. For directed graphs the construction (depicted in Figure 3 of~\cite{bernstein2025graphsshortestpathstructure}) has a very simple structure: $n/3$ layers where each layer is a directed 3-cycle, and each pair of adjacent layers is connected by a downwardly directed carefully chosen perfect matching of 3 edges. The undirected lower bound follows the same general structure but with 5-cycles instead of 3-cycles. In both constructions, the weights of the edges increase exponentially up the layers. The analysis is also quite simple: The graph is constructed so that for certain pairs of vertices in adjacent layers $i, i+1$ their shortest path consists of \emph{two} edges in layer $i$ and a cross-layer edge between layers $i$ and $i+1$, while a corresponding not-shortest path contains \emph{one} edge in layer $i+1$ and a cross-layer edge. After reweighting, the sum of edge weights on all of these shortest paths must be smaller sum of the edge weights on these corresponding not-shortest paths. Then, because the not-shortest paths have a half as many edges than the shortest paths (ignoring the cross-layer edges which cancel out), the average edge weight in layer $i+1$ must be twice the average edge weight in layer $i$. 

These arguments rely heavily on the \emph{symmetry} of the construction. Every edge is ``identical'' to the other edges in its layer: each participates in a symmetric structure of shortest and not-shortest paths, and each layer forms a symmetric structure (a cycle). In contrast, DAGs are inherently asymmetric: there can't be cycles and there must be source(s) and sink(s). For this reason, any natural adaptation of the above construction to the case of DAGs seems to suffer from asymmetry that prevents a clean doubling-at-each-layer argument. The asymmetry allows certain edges to absorb more weight than others which ends up allowing all weights to collapse to polynomially-bounded integer values. Furthermore, for certain DAGs that one can construct, this collapse is nontrivial, and seems to require reasoning about intricate chain reactions of edge decreases. 

Because of these issues, our exponential lower bound construction for DAGs with integer edge weights is completely different and significantly more involved than the known constructions for general graphs. 
Instead of trying to directly reason about all possible chain reactions of edge decreases in candidate lower bound constructions, which becomes messy, our approach is to reduce from a new seemingly unrelated problem.
In particular, we reduce from a somewhat strange problem concerning what we call \emph{3-topological order preservers}. From there, ChatGPT 5.5 Pro was able to provide ideas for getting an exponential lower bound for 3-topological order preservers (after not being able to solve our original problem). This, together with our reduction, implies an exponential lower bound for integer-weighted shortest-paths preservers in DAGs.

To define 3-topological order preservers, we first provide the definition of \emph{topological ordering} that we will be working with. 

\begin{definition}[Topological ordering of a DAG]
For a DAG $G = (V,E)$, we say $h: V \to \mathbb R$ is a \emph{topological ordering} of $G$ if for all $(u,v) \in E$, we have $h(u) +1 \le h(v)$.
\end{definition} 

Note that according to this definition, each vertex is assigned a \emph{real} number (not necessarily an integer), and multiple vertices can be assigned the same number.

We also need to define an \emph{additive triple} of DAGs. See \Cref{fig:additive_triples} for examples.

\begin{definition}[Additive triple of DAGs]
    Let $(G_1,G_2,G_3)$ be a triple of DAGs, where each $G_i = (V, E_i)$ is defined on the same vertex set $V$. We say $(G_1, G_2, G_3)$ is an \emph{additive triple} if there exists topological orderings $h_i : V \to \mathbb R$ of $G_i$ such that for all $v \in V$, $h_1(v) + h_2(v) = h_3(v)$.
\end{definition}

\begin{figure}[htbp]
\centering

\begin{subfigure}{0.95\textwidth}
\centering

\begin{tikzpicture}[
    vertex/.style={circle, fill=black, inner sep=2pt},
    >=stealth
]

\node at (-0.5,1) {$G_1:$};
\node[vertex] (u1) at (0,0) {};
\node[vertex] (v1) at (2,0) {};
\node[below=5pt] at (u1) {$u$};
\node[below=5pt] at (v1) {$v$};
\draw[-{Stealth[length=8pt,width=6pt]}] (u1) -- (v1);

\node at (4.5,1) {$G_2:$};
\node[vertex] (u2) at (5,0) {};
\node[vertex] (v2) at (7,0) {};
\node[below=5pt] at (u2) {$u$};
\node[below=5pt] at (v2) {$v$};
\draw[-{Stealth[length=8pt,width=6pt]}] (v2) -- (u2);

\node at (9.5,1) {$G_3:$};
\node[vertex] (u3) at (10,0) {};
\node[vertex] (v3) at (12,0) {};
\node[below=5pt] at (u3) {$u$};
\node[below=5pt] at (v3) {$v$};
\draw[-{Stealth[length=8pt,width=6pt]}] (u3) -- (v3);

\end{tikzpicture}

\caption{The triple $(G_1, G_2, G_3)$ of DAGs above is an additive triple. For example, we can take the topological orderings $h_1(u) = 0, h_1(v) = 5, h_2(u) = 1, h_2(v) = 0, h_3(u) = 1, h_3(v) = 5$.}

\end{subfigure}

\vspace{0.5cm}

\begin{subfigure}{0.95\textwidth}
\centering

\begin{tikzpicture}[
    vertex/.style={circle, fill=black, inner sep=2pt},
    >=stealth
]

\node at (-0.5,1) {$G_1:$};
\node[vertex] (u1) at (0,0) {};
\node[vertex] (v1) at (2,0) {};
\node[below=5pt] at (u1) {$u$};
\node[below=5pt] at (v1) {$v$};
\draw[-{Stealth[length=8pt,width=6pt]}] (u1) -- (v1);

\node at (4.5,1) {$G_2:$};
\node[vertex] (u2) at (5,0) {};
\node[vertex] (v2) at (7,0) {};
\node[below=5pt] at (u2) {$u$};
\node[below=5pt] at (v2) {$v$};
\draw[-{Stealth[length=8pt,width=6pt]}] (u2) -- (v2);

\node at (9.5,1) {$G_3:$};
\node[vertex] (u3) at (10,0) {};
\node[vertex] (v3) at (12,0) {};
\node[below=5pt] at (u3) {$u$};
\node[below=5pt] at (v3) {$v$};
\draw[-{Stealth[length=8pt,width=6pt]}] (v3) -- (u3);

\end{tikzpicture}

\caption{The triple $(G_1, G_2, G_3)$ of DAGs above is not an additive triple.}

\end{subfigure}

\caption{Examples of additive triples and non-additive triples.}

\label{fig:additive_triples}

\end{figure}

We focus on the following question: Given an additive triple of DAGs, do there always exist topological orderings $h_i$ of $G_i$ with $h_1 + h_2 = h_3$, such that the maximum range of all $h_i$ is polynomially bounded? More formally:

\begin{definition}[3-topological order preserver]
    Let $(G_1, G_2, G_3)$ be an additive triple of DAGs. We say $(h_1, h_2, h_3)$ is a \emph{3-topological order preserver} if $h_i$ is a topological ordering of $G_i$ for $i = 1,2,3$, and for all $v \in V$, $h_1(v) + h_2(v) = h_3(v)$.
\end{definition}

\begin{definition}[Range of a 3-topological order preserver]
    Let $(G_1, G_2, G_3)$ be an additive triple of DAGs with $G_i = (V, E_i)$, and let $(h_1, h_2, h_3)$ be a 3-topological order preserver. The \emph{range} of $(h_1, h_2, h_3)$ is defined as the quantity
    \[
    \mathrm{range}(h_1,h_2,h_3) := \max_{i \in \{1,2,3\}} \left(\max_{v \in V} h_i(v) - \min_{v \in V} h_i(v)\right)
    \]
\end{definition}

Formally, the above question becomes whether every additive triple of DAGs on $n$ vertices has a 3-topological order preserver with $\mathrm{poly}(n)$ range.

Now that we have defined 3-topological order preservers, we are ready to state our reduction from 3-topological order preservers to shortest-paths preservers. 

\begin{restatable}[Reduction from 3-topological order preservers to shortest-paths preservers] {theorem}{reduction}
\label{thm:reduction}
    Given an additive triple of DAGs $(G_1,G_2,G_3)$ with $G_i = (V,E_i)$, we can construct a 3-layered DAG $G^* = (V^* = V_1^* \cup V_2^* \cup V_3^*, E^*, w)$ with $|V_2^*| = 3$, such that 
    \begin{itemize}
        \item 
        $|V_1^*| = |V|$

        \item 
        $|V_3^*| = |E_1| + |E_2| + |E_3|$

        \item 
        The minimum possible range of a 3-topological order preserver of $(G_1,G_2,G_3)$ is at most 2 times the minimum possible maximum weight of a gap-1 shortest-paths preserver of $G^*$. That is, 
        \[\min_{(h_1',h_2',h_3')} \left(\mathrm{range}(h_1',h_2',h_3')\right) \le 2 \cdot \min_{G' = (V^*, E^*,w')} \left(\max_{e \in E^*} w'(e)\right)
        \]
        where the first minimum is over 3-topological order preservers $(h_1',h_2',h_3')$ of $(G_1,G_2,G_3)$ and the second minimum is over gap-1 shortest-paths preservers $G' = (V^*, E^*, w')$ of $G^*$.
    \end{itemize}
\end{restatable}

We refer the reader to \Cref{sec:reduction} for intuition about how the reduction works. 

Finally, with the help of ChatGPT 5.5 Pro, we construct an exponential lower bound for the range of 3-topological order preservers.

\begin{restatable}[Exponential lower bound for the range of 3-topological order preservers]{theorem}{threedaglb}\label{thm:three_dag_lb}

    There exists additive triple of DAGs $(G_1,G_2,G_3)$ where $G_i = (V, E_i)$ with $n = |V|$ and $m = |E_1| + |E_2| + |E_3|$, such that for all 3-topological order preservers $(h_1,h_2,h_3)$, we have
    \[
    \mathrm{range}(h_1,h_2,h_3) \ge 2^{\Omega(n+m)}.
    \]
\end{restatable}

Combining this with the reduction of \Cref{thm:reduction} proves our exponential lower bound for gap-1 shortest-paths preservers on DAGs.

To generalize our exponential lower bound to the $\alpha$-stretch case, we use a modification of our construction for the exact version. The structure of the graph remains the same, but the edge weights are carefully updated so that each shortest path in the original graph becomes the \emph{only} $\alpha$-approximate shortest path in the updated graph, forcing us to preserve the same set of paths as before. These weight updates were also done with the help of ChatGPT 5.5 Pro.

Finally, as a secondary result, we show in the appendix that our reduction is actually bidirectional. That is, 3-topological order preservers are essentially equivalent to gap-1 shortest-paths preservers in 3-layered DAGs with 3 vertices in the middle layer. Formally, we prove the following theorem.

\begin{restatable}[Reduction from shortest-paths preservers to 3-topological order preservers]{theorem}{otherdir}\label{thm:other_direction}
Given a 3-layered DAG $G = (V = V_1 \cup V_2 \cup V_3, E, w)$ with $|V_2| = 3$, we can construct an additive triple of DAGs $(G_1^*, G_2^*, G_3^*)$ with $G_i^* = (V^*, E_i^*)$, such that
    \begin{itemize}
        \item 
        $|V^*| = |V_1| + |V_3|$

        \item 
        $|E_1^*| + |E_2^*| + |E_3^*| = 2 \cdot |V_1| \cdot |V_3|$

        \item 
        The minimum possible maximum weight of a gap-1 shortest-paths preserver of $G$ is at most 1 plus the minimum possible range of a 3-topological order preserver of $(G_1^*, G_2^*, G_3^*)$. That is,
        \[
        \min_{G' = (V, E, w')} \left(\max_{e \in E} w'(e)\right) \le 1 + \min_{(h_1^*, h_2^*, h_3^*)} \left( \mathrm{range}(h_1^*, h_2^*, h_3^*)\right)
        \]
        where the first minimum is over gap-1 shortest-paths preservers $G' = (V, E, w')$ of $G$ and the second minimum is over 3-topological order preservers $(h_1^*, h_2^*, h_3^*)$ of $(G_1^*, G_2^*,G_3^*)$.
    \end{itemize}
\end{restatable}

The paper is organized as follows. In \Cref{sec:reduction}, we present the reduction from 3-topological order preservers to shortest-paths preservers (\Cref{thm:reduction}). In \Cref{sec:3DAG} we present the exponential lower bound for 3-topological order preservers (\Cref{thm:three_dag_lb}). Together, those two theorems prove the exponential lower bound for integer-weighted shortest-paths preservers (\Cref{thm:main}). In \Cref{sec:approx}, we  generalize the exponential lower bound to the $\alpha$-stretch setting (\Cref{thm:approx}). In \Cref{sec:ub} we prove the linear upper bound for integer-weighted shortest-paths preservers on 3-layered DAGs with 2 vertices in the middle layer (\Cref{thm:upper_bound}). Finally, in \Cref{sec:other}, we show the reduction in the other direction from shortest-paths preservers to 3-topological order preservers (\Cref{thm:other_direction}).

\section{Reduction from 3-Topological Order Preservers to Shortest-Paths Preservers}\label{sec:reduction}

We prove \cref{thm:reduction}:

\reduction*

\subsection{Construction}

We construct $G^* = (V^*, E^*, w)$ as follows (see \Cref{fig:three-layered-partition}): 
\begin{itemize}
    \item 
    $V^*$: Let $V_1^* := V$ be the entire original vertex set, $V_2^* := \{z_1,z_2,z_3\}$, and $V_3^* := V_{3,1}^* \cup V_{3,2}^* \cup V_{3,3}^*$ where $V_{3,i}^* := \{v_e: e \in E_i\}$, so $V_{3,i}^*$ contains one vertex $v_e$ for each original edge $e \in E_i$. 

    \item 
    $E^*$: Between $V_1^*$ and $V_2^*$, we add the complete set of edges $\{(u,z_1), (u,z_2),(u,z_3)\}$ for each $u \in V_1^*$. Between $V_2^*$ and $V_3^*$, we add edges depending on the vertex in $V_3^*$: for each $v_e \in V_{3,1}^*$, we add edges $\{(z_1,v_e), (z_2,v_e)\}$; for each $v_e \in V_{3,2}^*$, we add edges $\{(z_2,v_e), (z_3,v_e)\}$; for each $v_e \in V_{3,3}^*$, we add edges $\{(z_1,v_e), (z_3,v_e)\}$.

    \item 
    $w$: We define the weights based on the topological orderings $h_i$ of $G_i$ with $h_1 + h_2 = h_3$, which exist by additivity of $(G_1,G_2,G_3)$. 
    
    For $u \in V_1^*$, let $w(u,z_1) :=0$, $w(u,z_2) := -2h_1(u)$, and $w(u,z_3) := -2h_3(u)$.

    For $v_e \in V_3^*$, suppose $e = (u,v)$. If $v_e \in V_{3,1}^*$, then let $w(z_1,v_e) := 1$ and $w(z_2,v_e) := 2h_1(u) + 2$. If $v_e \in V_{3,2}^*$, then let $w(z_2,v_e) := 1$ and $w(z_3,v_e) := 2h_2(u) +2$. If $v_e \in V_{3,3}^*$, then let $w(z_1,v_e) := 1$ and $w(z_3,v_e) := 2h_3(u) + 2$. 

    Note that, depending on the values of $h_i$, some of the above weights could be negative. We claim that it is okay to have negative weights because we can add some $C > 0$ to all edge weights while preserving all shortest paths: the weight of each path of the form $a \to b \to c$ is increased by $2C$, so shortest paths remain shortest. 
\end{itemize}

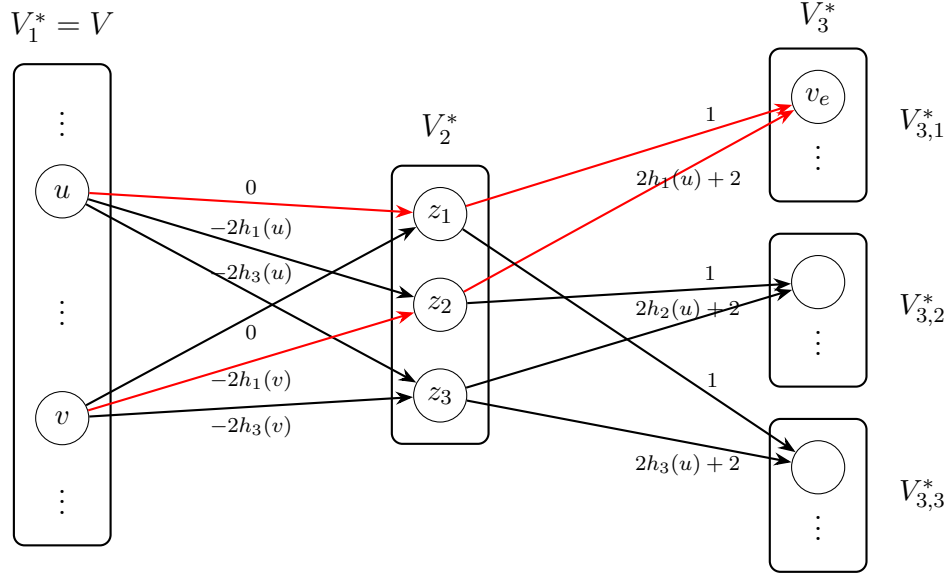
\begin{figure}[htbp]
\centering

\begin{tikzpicture}[
vertex/.style={circle, draw, inner sep=1.2pt, minimum size=20pt},
blackedge/.style={->, >=Stealth, thick, draw=black},
rededge/.style={->, >=Stealth, thick, draw=red},
setbox/.style={draw, rounded corners, thick, inner sep=8pt},
layerlabel/.style={font=\large},
partlabel/.style={font=\normalsize},
pad/.style={inner sep=0pt, minimum size=0pt},
weightlabel/.style={font=\scriptsize, inner sep=0pt}
]

\def\xone{-5}
\def\xtwo{0}
\def\xthree{5}

\node (V1topdots) at (\xone,2.5) {$\vdots$};
\node[vertex] (u) at (\xone,1.5) {$u$};
\node (V1middots) at (\xone,0.0) {$\vdots$};
\node[vertex] (v) at (\xone,-1.5) {$v$};
\node (V1botdots) at (\xone,-2.5) {$\vdots$};

\node[setbox, fit=(V1topdots)(u)(V1middots)(v)(V1botdots)] (V1box) {};
\node[layerlabel, above=4pt of V1box] {$V_1^* = V$};

\node[vertex] (z1) at (\xtwo,1.2) {$z_1$};
\node[vertex] (z2) at (\xtwo,0.0) {$z_2$};
\node[vertex] (z3) at (\xtwo,-1.2) {$z_3$};

\node[setbox, fit=(z1)(z2)(z3)] (V2box) {};
\node[layerlabel, above=4pt of V2box] {$V_2^*$};

\node[pad]    (V31top) at (\xthree,3.10) {};
\node[vertex] (ve)     at (\xthree,2.75) {$v_e$};
\node         (V3dots1) at (\xthree,2.05) {$\vdots$};
\node[pad]    (V31bot) at (\xthree,1.70) {};
\node[setbox, fit=(V31top)(ve)(V3dots1)(V31bot)] (V3box1) {};
\node[partlabel, right=8pt of V3box1] {$V_{3,1}^*$};

\node[pad]    (V32top) at (\xthree,0.65) {};
\node[vertex] (w2)     at (\xthree,0.30) {};
\node         (V3dots2) at (\xthree,-0.40) {$\vdots$};
\node[pad]    (V32bot) at (\xthree,-0.75) {};
\node[setbox, fit=(V32top)(w2)(V3dots2)(V32bot)] (V3box2) {};
\node[partlabel, right=8pt of V3box2] {$V_{3,2}^*$};

\node[pad]    (V33top) at (\xthree,-1.80) {};
\node[vertex] (w3)     at (\xthree,-2.15) {};
\node         (V3dots3) at (\xthree,-2.85) {$\vdots$};
\node[pad]    (V33bot) at (\xthree,-3.20) {};
\node[setbox, fit=(V33top)(w3)(V3dots3)(V33bot)] (V3box3) {};
\node[partlabel, right=8pt of V3box3] {$V_{3,3}^*$};

\node[layerlabel] at (\xthree,3.80) {$V_3^*$};

\draw[rededge] (u) --
    node[weightlabel, pos=0.5, above=3pt] {$0$}
    (z1);

\draw[blackedge] (u) --
    node[weightlabel, pos=0.5, above=3pt] {$-2h_1(u)$}
    (z2);

\draw[blackedge] (u) --
    node[weightlabel, pos=0.5, above=3pt] {$-2h_3(u)$}
    (z3);

\draw[blackedge] (v) --
    node[weightlabel, pos=0.5, below=3pt] {$0$}
    (z1);

\draw[rededge] (v) --
    node[weightlabel, pos=0.5, below=3pt] {$-2h_1(v)$}
    (z2);

\draw[blackedge] (v) --
    node[weightlabel, pos=0.5, below=3pt] {$-2h_3(v)$}
    (z3);

\draw[rededge] (z1) --
    node[weightlabel, pos=0.75, above=3pt] {$1$}
    (ve);

\draw[rededge] (z2) --
    node[weightlabel, pos=0.68, below=0pt] {$2h_1(u)+2$}
    (ve);

\draw[blackedge] (z2) --
    node[weightlabel, pos=0.75, above=3pt] {$1$}
    (w2);

\draw[blackedge] (z3) --
    node[weightlabel, pos=0.68, above=0pt] {$2h_2(u)+2$}
    (w2);

\draw[blackedge] (z1) --
    node[weightlabel, pos=0.75, above=3pt] {$1$}
    (w3);

\draw[blackedge] (z3) --
    node[weightlabel, pos=0.68, below=4pt] {$2h_3(u)+2$}
    (w3);

\end{tikzpicture}

\caption{The construction of $G^*$. For $v_e \in V_{3,1}^*$ with $e = (u,v)$, the red paths are the unique shortest paths from $u$ to $v_e$ and from $v$ to $v_e$. To make all edge weights positive, imagine adding a sufficiently large $C > 0$ to all edge weights. }
\label{fig:three-layered-partition}

\end{figure}

\subsection{Intuition Behind the Construction}
\label{subsec:intuition}

The weight assignment may seem unintuitive at first, so we provide a brief roadmap of the structure of the proof. The formal analysis is self-contained within \Cref{subsec:an}.

Ultimately we wish to establish a relationship between the range of the existing topological orders $h_i$ of $(G_1,G_2,G_3)$ and the edge weights $w$ in $G^*$, so that a similar relationship can be established between the \emph{minimum-range} 3-topological order preserver $h'_i$ of $(G_1,G_2,G_3)$  and the shortest path preserver of $G^*$ with \emph{minimized maximum weight}. 
This way, we will be able to show that a valid shortest-path preserving reweighting $w'$ of $G^*$ can be used to construct valid topological orderings $h'_i$. But how do we establish such a relationship? 

First, the fact that certain paths are shortest and the rest are \emph{not} shortest implies a set of inequalities on the reweighted edge weights of $G^*$. 
The $h'_i$ values will then be chosen based on the reweighted edge weights $w'$ of various edges so that substituting the $h'_i$ values into the shortest-path-based inequalities implies that the $h'_i$ are valid topological orders with $h'_1+h'_2=h'_3$. 

In particular, if $a \rightarrow b \rightarrow c$ is a shortest path from $a\in V^*_1$ to $c\in V^*_3$ and $a \rightarrow b' \rightarrow c$ is \emph{not} a shortest path, then we know that $w(a,b)+w(b,c)<w(a,b')+w(b',c)$, or equivalently, $w(a,b) - w(a,b') < w(b',c) - w(b,c)$. 
Writing the inequality in the second way reveals two \emph{differences} between edge weights that are of interest. The construction is such that the first difference $(w(a,b)-w(a,b'))$ captures the $h'_i$ values according to the following equations, established in \cref{lem:h_i'}: for each $u \in V$, we let 
\[
h_1'(u) := w'(u,z_1) - w'(u,z_2), \quad h_2'(u) := w'(u,z_2) - w'(u,z_3), \quad h_3'(u) := w'(u,z_1) - w'(u,z_3).
\]
Note that these equations indeed imply the desired fact that $h'_1+h'_2=h'_3$ due the cancellation of $w'(u,z_2)$. 

This is the motivation behind the following weights specified in the construction: $w(u,z_1) :=0$, $w(u,z_2) := -2h_1(u)$, and $w(u,z_3) := -2h_3(u)$. With these weights, indeed the above equations are satisfied for $w$ and $h_i$, but with a factor 2 of slack for convenience. 
However, this alone isn't enough because we also want these equations to hold (up to constant factors) 
for \emph{any} reweighting $w'$. For this we need the shortest path structure to be ``rigid'' enough. The structure that we will enforce is given by \cref{lem:sp_structure} and paraphrased here:

Given an edge $(u,v)\in E_1$, the shortest path from the copy of $u$ (respectively $v$) in $V^*_1$ to the vertex $v_e$ representing edge $e=(u,v)$ in $V^*_{3,1}$  goes through $z_1$ (respectively $z_2$). Similarly, for $V^*_{3,2}$ and $V^*_{3,3}$, the corresponding paths go through $z_2,z_3$ and $z_1,z_3$, respectively.

With this shortest path structure, as per the above equations, the construction assigns the first difference (corresponding to $(w(a,b)-w(a,b'))$) to be $2h_1(u)$ for the shortest path $u\rightarrow z_1 \rightarrow v_e$ compared to the not-shortest path $u\rightarrow z_2 \rightarrow v_e$, and to be $2h_1(v)$ for the shortest path $v\rightarrow z_2 \rightarrow v_e$ compared to the not-shortest path $v\rightarrow z_1 \rightarrow v_e$ (and analogously for $h_2$ and $h_3$). To respect the desired shortest path structure we need to consider the second difference (corresponding to $(w(b',c)-w(b,c))$) for both of those shortest/not-shortest path pairs. In particular, the second difference for these pairs of paths are both $w(z_2,v_e) - w(z_1,v_e)$. 
For this reason, to respect the specified shortest path structure, we need the second difference to be \emph{between} the two first differences $2h_1(u)$ and $2h_1(v)$. 

To achieve this, the construction sets the second difference to $2h_1(u)+1$, which is clearly above $2h_1(u)$ but is below $2h_1(v)$ because $(u,v)$ is an edge in $G_1$, and $h_1$ is a valid topological order of $G_1$.  This line of reasoning is why it is useful for the vertices in $V^*_3$ to represent \emph{edges} each $G_i$, while the vertices in $V^*_1$ represent vertices of the $G_i$'s. Finally, to ensure that the second difference is indeed $2h_1(u)+1$, the construction assigns $w(z_1,v_e) = 1$ and $w(z_2,v_e) = 2h_1(u) + 2$ so that $w(z_2,v_e)-w(z_1,v_e)=2h_1(u) + 1$, and analogously for $h_2$ and $h_3$.

\subsection{Analysis}\label{subsec:an}

We prove some properties about $G^*$: 

\begin{lemma}\label{lem:differences}
    For each $u \in V$, we have 
    \[
    2h_1(u) = w(u,z_1) - w(u,z_2), \quad 2h_2(u) = w(u,z_2) - w(u,z_3), \quad 2h_3(u) = w(u,z_1) - w(u,z_3).
    \]
    
\end{lemma}

\begin{proof}
    This follows from the definition of $w$, using the fact that $h_1 + h_2 = h_3$: we have
    \[
    w(u,z_1) - w(u,z_2) = 0 - (-2h_1(u)) = 2h_1(u),
    \]
    \[
    w(u,z_2) - w(u,z_3) = 2(h_3(u) - h_1(u)) = 2h_2(u),
    \]
    \[
    w(u,z_1) - w(u,z_3) = 0 - (-2h_3(u)) = 2h_3(u).
    \]
    Note that adding some $C$ to all edge weights doesn't change the difference between 2 edge weights, so \cref{lem:differences} still holds after we add $C$ to all edge weights.
\end{proof}

\begin{lemma}\label{lem:sp_structure}
    $G^*$ has the following shortest path structure:
    \begin{itemize}
        \item 
        For each $v_e \in V_{3,1}^*$ with $e = (u,v)$, the unique shortest path from $u$ to $v_e$ is $u \to z_1 \to v_e$, and the unique shortest path from $v$ to $v_e$ is $v \to z_2 \to v_e$. 

        \item 
        For each $v_e \in V_{3,2}^*$ with $e = (u,v)$, the unique shortest path from $u$ to $v_e$ is $u \to z_2 \to v_e$, and the unique shortest path from $v$ to $v_e$ is $v \to z_3 \to v_e$.

        \item 
        For each $v_e \in V_{3,3}^*$ with $e = (u,v)$, the unique shortest path from $u$ to $v_e$ is $u \to z_1 \to v_e$, and the unique shortest path from $v$ to $v_e$ is $v \to z_3 \to v_e$.
    \end{itemize}
\end{lemma}

\begin{proof}
    Suppose $v_e \in V_{3,1}^*$ with $e = (u,v)$. Recall that $w(z_2,v_e) = 2h_1(u) + 2$ and $w(z_1,v_e) = 1$, so we have 
    \begin{equation}
        w(z_2,v_e) - w(z_1,v_e) = 2h_1(u) + 1 
        \label{eq:ve_difference}
    \end{equation}
    We first look at the shortest path from $u$ to $v_e$. Since the only paths between $u$ and $v_e$ are $u \to z_1 \to v_e$ and $u \to z_2 \to v_e$, it is enough to check that $w(u \to z_2 \to v_e) > w(u \to z_1 \to v_e)$. We have 
    \begin{align*}
        w(u \to z_2 \to v_e) - w(u \to z_1 \to v_e) &= (w(z_2,v_e) - w(z_1,v_e)) - (w(u,z_1) - w(u,z_2)) \\
            &= (2h_1(u) + 1) - 2h_1(u) \qquad (\text{by \cref{eq:ve_difference} and \cref{lem:differences}}) \\ 
        &= 1.
    \end{align*}
    Using terms from the discussion in \cref{subsec:intuition}, this is saying the second difference at $u$ is at least $1$ greater than the first difference at $u$.
    
    For the shortest path from $v$ to $v_e$, it is enough to check that $w(v \to z_1 \to v_e) > w(v \to z_2 \to v_e)$. We have 
    \begin{align*}
        w(v \to z_1 \to v_e) - w(v \to z_2 \to v_e) &= (w(v,z_1) - w(v,z_2)) - (w(z_2,v_e) - w(z_1,v_e)) \\ 
        &= 2h_1(v) - (2h_1(u) + 1) \qquad (\text{by \cref{eq:ve_difference} and \cref{lem:differences}}) \\
        &\ge 1,
    \end{align*}
    where $2(h_1(v) - h_1(u)) \ge 2$ because $h_1$ is a topological ordering of $G_1$, and $(u,v)$ is an edge in $G_1$. Similarly, using terms from the discussion in \cref{subsec:intuition}, this is saying the first difference at $v$ is at least 1 greater than the second difference at $v$.

    The proof is exactly the same for $v_e \in V_{3,2}^*, V_{3,3}^*$. For $v_e \in V_{3,2}^*$, we replace all occurrences of $z_1$ by $z_2$, all occurrences of $z_2$ by $z_3$, and all occurrences of $h_1$ by $h_2$ in the above argument. Then \cref{eq:ve_difference} still holds since we set the weights $w(z_2,v_e) = 1$ and $w(z_3,v_e) = 2h_2(u) + 2$, and we use the second difference in \cref{lem:differences} instead of the first difference. For $v_e \in V_{3,3}^*$, we keep all occurrences of $z_1$, replace all occurrences of $z_2$ by $z_3$, and replace all occurrences of $h_1$ by $h_3$. 
\end{proof}

Now, we prove the key fact that any gap-1 reweighting $w'$ of $G^*$ that preserves all shortest paths can be used to construct a 3-topological order preserver $(h_1',h_2',h_3')$.

Let $G' = (V^*, E^*, w')$ be any gap-1 shortest-paths preserver of $G^*$. We define $h_i': V \to \mathbb Z$ as follows: for each $u \in V$,
\[
h_1'(u) := w'(u,z_1) - w'(u,z_2), \quad h_2'(u) := w'(u,z_2) - w'(u,z_3), \quad h_3'(u) := w'(u,z_1) - w'(u,z_3).
\]

\begin{lemma}\label{lem:h_i'}
    $(h_1',h_2',h_3')$ is a 3-topological order preserver of $(G_1,G_2,G_3)$. That is, $h_i'$ is a topological ordering of $G_i$, and $h_1' + h_2' = h_3'$.
\end{lemma}

\begin{proof}
    It is clear from the definition that $h_1' + h_2' = h_3'$. We show that $h_1'$ is a topological ordering of $G_1$. Suppose $e = (u,v) \in E_1$ and consider $v_e \in V_{3,1}^*$. Since $G'$ is gap-1 shortest-paths preserving, by \cref{lem:sp_structure} the unique shortest path in $G'$ from $u$ to $v_e$ is $u \to z_1 \to v_e$, so $w'(u \to z_1 \to v_e) + 1\le  w'(u \to z_2 \to v_e)$, which expands to
    \[
    h_1'(u) = w'(u,z_1) - w'(u,z_2) \le w'(z_2,v_e) - w'(z_1,v_e)-1.
    \]
    Also, by \cref{lem:sp_structure} the unique shortest path in $G'$ from $v$ to $v_e$ is $v \to z_2 \to v_e$, so $w'(v \to z_2 \to v_e) +1\le w'(v \to z_1 \to v_e)$, which expands to
    \[
    h_1'(v) = w'(v,z_1) - w'(v,z_2) \ge w'(z_2,v_e) - w'(z_1,v_e)+1.
    \]
    So $h_1'(u) +2\le h_1'(v)$, which proves $h_1'$ is a topological ordering of $G_1$.

    The argument is exactly the same for $h_2',h_3'$. To show $h_2'$ is a topological ordering of $G_2$, we look at $e = (u,v) \in E_2$ with $v_e \in V_{3,2}^*$. Using what \cref{lem:sp_structure} says about the shortest path structure of $v_e$, we can use the above argument with all occurrences of $h_1'$ replaced by $h_2'$, and $z_1, z_2$ replaced by $z_2,z_3$, respectively. For $h_3'$, by considering $e = (u,v) \in E_3$ and $v_e \in V_{3,3}^*$, we can use the above argument with all occurrences of $h_1'$ replaced by $h_3'$, and all occurrences of $z_2$ replaced by $z_3$. 
\end{proof}

\begin{lemma}\label{lem:bounded_range}
    We have 
    \[\mathrm{range}(h_1',h_2',h_3') \le 2 \cdot \max_{e \in E^*} w'(e).
    \]
\end{lemma}

\begin{proof}
    For any $i \in \{1,2,3\}$ and $u,v \in V$, we know that for indices 
    \[
    (j,k) = \begin{cases}
        (1,2) & i = 1 \\
        (2,3) & i = 2 \\
        (1,3) & i = 3
    \end{cases}
    \]
    we have
    \[
    h_i'(v) - h_i'(u) = w'(v,z_j) + w'(u,z_k) - w'(v,z_k) - w'(u,z_j)\le 2 \cdot \max_{e \in E^*} w'(e).
    \]
    Since the above holds for all $i \in \{1,2,3\}$ and $u,v \in V$, we have
    \[
    \mathrm{range}(h_1',h_2',h_3') = \max_{i \in \{1,2,3\}} \left(\max_{v \in V} h_i'(v) - \min_{v \in V} h_i'(v) \right) \le 2 \cdot \max_{e \in E^*} w'(e).
    \]
\end{proof}

Recall that the original inequality we are trying to prove is
\begin{equation}
    \min_{(h_1',h_2',h_3')} \left(\mathrm{range}(h_1',h_2',h_3')\right) \le 2 \cdot \min_{G' = (V^*, E^*,w')} \left(\max_{e \in E^*} w'(e)\right)
    \label{ineq:original}
\end{equation}
where the first minimum is over 3-topological order preservers $(h_1',h_2',h_3')$ of $(G_1,G_2,G_3)$ and the second minimum is over gap-1 shortest-paths preservers $G' = (V^*, E^*, w')$ of $G^*$. 

In \cref{lem:bounded_range} we showed that for any gap-1 shortest-paths preserving reweighting $w'$, there exists $(h_1',h_2',h_3')$ with 
\[
\mathrm{range}(h_1',h_2',h_3') \le 2 \cdot \max_{e \in E^*} w'(e).
\]
In particular, this inequality holds for the reweighting $w'$ that minimizes $\max_{e \in E^*} w'(e)$. Since the left hand side of \cref{ineq:original} takes a minimum over all 3-topological order preservers and $(h_1',h_2',h_3')$ is a 3-topological order preserver by \cref{lem:h_i'}, we have proved \cref{ineq:original}.

This completes the proof of \cref{thm:reduction}.

\section{Lower Bound for 3-Topological Order Preservers}\label{sec:3DAG}

We prove \cref{thm:three_dag_lb}:

\threedaglb*

Before proving \cref{thm:three_dag_lb}, we will show that \cref{thm:three_dag_lb} implies \cref{thm:main} using the reduction from \cref{sec:reduction}. Recall \cref{thm:main}:

\maintheorem*

\begin{proof}
    Let $(G_1,G_2,G_3)$ be an exponential lower bound additive triple of DAGs from \cref{thm:three_dag_lb}, with $G_i = (V,E_i)$, $n = |V|$, and $m = |E_1| + |E_2| + |E_3|$. By \cref{thm:reduction}, we can construct a 3-layered DAG $G^* = (V^* = V_1^* \cup V_2^* \cup V_3^*, E^*, w)$ with $|V_1^*| = n$, $|V_2^*| = 3$, $|V_3^*| = m$, such that for any shortest-paths preserving edge reweighting $w'$ of $G^*$, 
    \begin{align*}
        \max_{e \in E^*} w'(e) &\ge \frac 12 \cdot \min_{(h_1',h_2',h_3') \text{ 3-topological order preserver of } (G_1,G_2,G_3)}  \left(\mathrm{range}(h_1',h_2',h_3')\right) \\
        &\ge 2^{\Omega(n+m)} \\
        &= 2^{\Omega(|V^*|)}.
    \end{align*}
\end{proof}

Now we will prove \cref{thm:three_dag_lb}.

\subsection{Construction}\label{subsec:3dag_construction}

Assume $n = 2 + 4t$ for some $t \ge 0$. We define a triple of DAGs $(G_1, G_2, G_3)$ where $G_i = (V, E_i)$ as follows:

\begin{itemize}
    \item 
    The common vertex set is 
    \[
    V := \{b_0,c_0\} \cup \bigcup_{r=1}^t \{a_r,b_r,c_r,d_r\},
    \]
    where level 0 consists of 2 initial vertices $\{b_0,c_0\}$, and $t$ additional levels where level $r \in [t]$ contains the vertices $\{a_r,b_r,c_r,d_r\}$. 

    \item 
    Edges in all $E_i$ either go between vertices in the same level, or go between adjacent levels. The edge sets are shown in \Cref{fig:three-graphs}. Formally, the edge sets are defined as follows:
    \[
    E_1 := \{(b_0,c_0)\} \cup \bigcup_{r\le t\text{ odd}} \{(a_{r},b_{r-1}),(c_{r},b_{r-1}),(c_{r-1},b_r),(c_{r-1},d_r)\} \cup \bigcup_{0 < r \le t \text{ even}} \{(a_r,d_r)\}
    \]
    \[
    E_2 := \bigcup_{0 < r \le t \text{ even}} \{(a_{r},b_{r-1}),(c_{r},b_{r-1}),(c_{r-1},b_r),(c_{r-1},d_r)\} \cup \bigcup_{r\le t \text{ odd}} \{(a_r,d_r)\}
    \]
    \[
    E_3 := \bigcup_{r=1}^t \{(b_r,a_r),(d_r,c_r)\}.
    \]
\end{itemize}

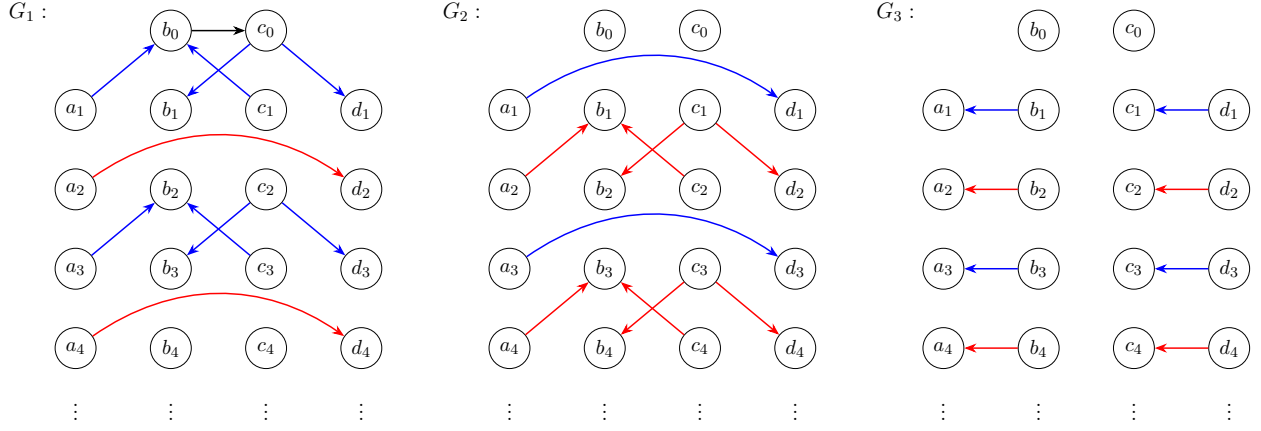
\begin{figure}[htbp]
\centering

\resizebox{\textwidth}{!}{%
\begin{tikzpicture}[
vertex/.style={circle, draw, inner sep=1.2pt, minimum size=22pt},
edge/.style={->, >=Stealth, thick},
blackedge/.style={edge, black},
blueedge/.style={edge, blue},
rededge/.style={edge, red},
graphlabel/.style={font=\large}
]

\def\xa{0}
\def\xb{1.8}
\def\xc{3.6}
\def\xd{5.4}

\def\yzero{0}
\def\yone{-1.5}
\def\ytwo{-3.0}
\def\ythree{-4.5}
\def\yfour{-6.0}
\def\ydots{-7.1}

\def\xsep{8.2}

\newcommand{\PlaceVertices}[1]{%
\foreach \name/\x/\y/\lab in {
b0/\xb/\yzero/$b_0$,
c0/\xc/\yzero/$c_0$,
a1/\xa/\yone/$a_1$,
b1/\xb/\yone/$b_1$,
c1/\xc/\yone/$c_1$,
d1/\xd/\yone/$d_1$,
a2/\xa/\ytwo/$a_2$,
b2/\xb/\ytwo/$b_2$,
c2/\xc/\ytwo/$c_2$,
d2/\xd/\ytwo/$d_2$,
a3/\xa/\ythree/$a_3$,
b3/\xb/\ythree/$b_3$,
c3/\xc/\ythree/$c_3$,
d3/\xd/\ythree/$d_3$,
a4/\xa/\yfour/$a_4$,
b4/\xb/\yfour/$b_4$,
c4/\xc/\yfour/$c_4$,
d4/\xd/\yfour/$d_4$
}{
\node[vertex] (#1\name) at (\x,\y) {\lab};
}
\foreach \x in {\xa,\xb,\xc,\xd}{
\node at (\x,\ydots) {$\vdots$};
}
}

\newcommand{\EdgesGOne}[1]{%
\draw[blackedge] (#1b0) -- (#1c0);
\draw[blueedge]  (#1a1) -- (#1b0);
\draw[blueedge]  (#1c1) -- (#1b0);
\draw[blueedge]  (#1c0) -- (#1b1);
\draw[blueedge]  (#1c0) -- (#1d1);
\draw[rededge]   (#1a2) to[out=35,in=145] (#1d2);

```
\draw[blueedge]  (#1a3) -- (#1b2);
\draw[blueedge]  (#1c3) -- (#1b2);
\draw[blueedge]  (#1c2) -- (#1b3);
\draw[blueedge]  (#1c2) -- (#1d3);

\draw[rededge]   (#1a4) to[out=35,in=145] (#1d4);
```

}

\newcommand{\EdgesGTwo}[1]{%
\draw[blueedge]  (#1a1) to[out=35,in=145] (#1d1);
\draw[rededge]   (#1a2) -- (#1b1);
\draw[rededge]   (#1c2) -- (#1b1);
\draw[rededge]   (#1c1) -- (#1b2);
\draw[rededge]   (#1c1) -- (#1d2);

```
\draw[blueedge]  (#1a3) to[out=35,in=145] (#1d3);

\draw[rededge]   (#1a4) -- (#1b3);
\draw[rededge]   (#1c4) -- (#1b3);
\draw[rededge]   (#1c3) -- (#1b4);
\draw[rededge]   (#1c3) -- (#1d4);
```

}

\newcommand{\EdgesGThree}[1]{%
\draw[blueedge]  (#1b1) -- (#1a1);
\draw[blueedge]  (#1d1) -- (#1c1);
\draw[rededge]   (#1b2) -- (#1a2);
\draw[rededge]   (#1d2) -- (#1c2);

```
\draw[blueedge]  (#1b3) -- (#1a3);
\draw[blueedge]  (#1d3) -- (#1c3);

\draw[rededge]   (#1b4) -- (#1a4);
\draw[rededge]   (#1d4) -- (#1c4);
```

}

\newcommand{\GraphBlock}[4]{%
\begin{scope}[shift={(#1,0)}]
\node[graphlabel] at (-0.9,0.35) {$#2:$};
\PlaceVertices{#3}
#4{#3}
\end{scope}
}

\GraphBlock{0}{G_1}{G1}{\EdgesGOne}
\GraphBlock{\xsep}{G_2}{G2}{\EdgesGTwo}
\GraphBlock{2*\xsep}{G_3}{G3}{\EdgesGThree}

\end{tikzpicture}%
}

\caption{The first few levels of $(G_1,G_2,G_3)$. Future levels continue to alternate between the blue and red patterns.}
\label{fig:three-graphs}

\end{figure}

\subsection{Analysis}

First, we need to show that $(G_1,G_2,G_3)$ is an additive triple of DAGs. 

\begin{lemma}\label{lem:h_i}
    $(G_1,G_2,G_3)$ is an additive triple of DAGs. That is, there exists $h_1,h_2,h_3: V \to \mathbb Z$ such that each $h_i$ is a topological ordering of $G_i$, and $h_1 + h_2 = h_3$.
\end{lemma}

\begin{proof}
    We define $h_1$, $h_2$ as follows:
    \begin{itemize}
        \item 
        For $r$ odd, define
        \[
        h_1(a_r) = 0, \quad h_1(b_r) = 5^r-2, \quad h_1(c_r) = 0, \quad h_1(d_r) = 5^r-2
        \]
        and
        \[
        h_2(a_r) = 5^r, \quad h_2(b_r) = 1, \quad h_2(c_r) = 2 \cdot 5^r, \quad h_2(d_r) = 5^r+1.
        \]

        \item 
        For $r$ even, we swap the mappings above. Define
        \[
        h_1(a_r) = 5^r, \quad h_1(b_r) = 1, \quad h_1(c_r) = 2 \cdot 5^r, \quad h_1(d_r) = 5^r+1
        \]
        and
        \[
        h_2(a_r) = 0, \quad h_2(b_r) = 5^r-2, \quad h_2(c_r) = 0, \quad h_2(d_r) = 5^r-2.
        \]
    \end{itemize}

    So for all $r \ge 1$, $h_3 = h_1 + h_2$ satisfies
    \[
    h_3(a_r) = 5^r, \quad h_3(b_r) = 5^r-1, \quad h_3(c_r) = 2 \cdot 5^r, \quad h_3(d_r) = 2 \cdot 5^r -1.
    \]

    We check that each $h_i$ is a valid topological ordering for $G_i$.
    \begin{itemize}
        \item 
        $G_1$: For the initial edge we have $h_1(b_0) = 1 < 2 = h_1(c_0)$. For the remaining edges, we check that for all odd $r \le t$, 
        \[
        h_1(a_r) = h_1(c_r) = 0 < 1 = h_1(b_{r-1})
        \]
        \[
        h_1(c_{r-1}) = 2 \cdot 5^{r-1} < 5^{r}-2 = h_1(b_{r}) = h_1(d_{r}),
        \]
        and for all even $0 < r \le t$,
        \[
        h_1(a_{r}) = 5^r < 5^r+1 = h_1(d_{r}).
        \]

        \item 
        $G_2$: For all odd $r \le t$, we have
        \[
        h_2(a_{r}) = 5^{r} < 5^{r}+1 = h_2(d_{r}),
        \]
        and for all even $0 < r \le t$, we have
        \[
        h_2(a_{r}) = h_2(c_{r})  = 0 < 1 = h_2(b_{r-1})
        \]
        \[
        h_2(c_{r-1}) = 2 \cdot 5^{r-1} < 5^{r} - 2 = h_2(b_{r}) = h_2(d_{r}).
        \]

        \item 
        $G_3$: For all $1 \le r \le t$, we have 
        \[
        h_3(b_r) = 5^r-1 < 5^r = h_3(a_r)
        \]
        \[h_3(d_r) = 2 \cdot 5^r - 1 < 2 \cdot 5^r = h_3(c_r).
        \]
    \end{itemize}
    
\end{proof}

It remains to show that for any 3-topological order preserver $(h_1',h_2',h_3')$, there exists $i \in \{1,2,3\}$ and $u,v \in V$ such that $h_i'(v) - h_i'(u) \ge 2^{\Omega(n+m)}$. We make the following key claim:

\begin{lemma}\label{lem:doubling}
    For all odd $r \ge 1$, we have 
    \[
    h_2'(c_r) - h_2'(b_r) \ge 2 \cdot (h_1'(c_{r-1}) - h_1'(b_{r-1})),
    \]
    and for all even $r \ge 2$, we have 
    \[
    h_1'(c_r) - h_1'(b_r) \ge 2 \cdot (h_2'(c_{r-1}) - h_2'(b_{r-1})).
    \]
\end{lemma}

\begin{proof}
    We first look at the case when $r$ is odd. Let $D_{r-1} := h_1'(c_{r-1}) - h_1'(b_{r-1})$. Using the edges $(a_r,b_{r-1}), (c_{r-1},b_r) \in E_1$, we have 
    \begin{align*}
    h_1'(b_r) - h_1'(a_r) &= (h_1'(b_r) - h_1'(c_{r-1})) + (h_1'(c_{r-1}) - h_1'(b_{r-1})) +( h_1'(b_{r-1}) - h_1'(a_r)) \\
    &\ge 1 + D_{r-1} + 1 \\
    &> D_{r-1}.
    \end{align*}
    Similarly, using the edges $(c_r,b_{r-1}),(c_{r-1},d_r) \in E_1$, we have
    \[
    h_1'(d_r) - h_1'(c_r) > D_{r-1}.
    \]
    Now, using the edge $(b_r,a_r) \in E_3$, we have $h_3'(a_r) > h_3'(b_r)$, so
    \[
    (h_1' + h_2')(a_r) > (h_1' + h_2')(b_r) \implies h_2'(a_r) - h_2'(b_r) > h_1'(b_r) - h_1'(a_r) > D_{r-1}.
    \]
    Similarly, using the edge $(d_r,c_r) \in E_3$, we get 
    \[
    h_1'(c_r) - h_1'(d_r) > D_{r-1}.
    \]
    Finally, using the edge $(a_r,d_r) \in E_2$ and the inequalities above, we have
    \begin{align*}
        h_2'(c_r) - h_2'(b_r) &= (h_2'(c_r) - h_2'(d_r)) + (h_2'(d_r) - h_2'(a_r)) + (h_2'(a_r) - h_2'(b_r)) \\
        &\ge D_{r-1} + 1 + D_{r-1} \\
        &> 2D_{r-1}.
    \end{align*}

    For the case $r$ even, we can repeat the argument above but with $h_1'$ and $h_2'$ swapped, since the edges are the same except that the edges in $E_1$ are now in $E_2$, and the edges in $E_2$ are now in $E_1$. The only asymmetry is that the initial edge $(b_0,c_0)$ is only in $E_1$, but our argument above does not use that edge. 
\end{proof}

To complete the proof of \cref{thm:three_dag_lb}, we will show
\[\max\{h_1'(c_t) - h_1'(b_t), h_2'(c_t) - h_2'(b_t)\} \ge 2^t,
\]
which is enough because $t = \Omega(n + m)$ since each $G_i$ is clearly sparse. 

From the initial edge $(b_0,c_0) \in E_1$, we have $h_1'(c_0) - h_1'(b_0) \ge 1$. Then, we can repeatedly apply \cref{lem:doubling} to get $h_2'(c_r) - h_2'(b_r) \ge 2^r$ for all odd $r \ge 1$, and $h_1'(c_r) - h_1'(b_r) \ge 2^r$ for all even $r \ge 0$.

Therefore, either $h_1'(c_t) - h_1'(b_t) \ge 2^t$ or $h_2'(c_t) - h_2'(b_t) \ge 2^t$, depending on the parity of $t$. This finishes the proof of \cref{thm:three_dag_lb}.

\section{Approximate Version}\label{sec:approx}

The previous sections prove \cref{thm:main}, an exponential lower bound for exact shortest-paths preservers, by showing that 3-topological order preservers reduce to shortest-paths preservers (\cref{thm:reduction} and giving an exponential lower bound for 3-topological order preservers (\cref{thm:three_dag_lb}). We now turn our attention to  $\alpha$-approximate shortest-paths preservers, and show that the exponential lower bound extends even when arbitrary stretch is allowed.

\approx*

We first outline the intuition and how the $\alpha$-stretch version compares to the exact version.
Currently, the lower bound 3-layered DAG $G^*$ from \cref{thm:main} 
relies on a collection of shortest paths $\mathcal P$ (the shortest paths in \cref{lem:sp_structure}), such that any reweighting that preserves the shortest paths in $\mathcal P$ has exponential maximum weight. One type of shortest path in $\mathcal P$ is of the form $P := u \to z_1 \to v_e$, where $v_e$ was created from $e = (u,v) \in E_1$ (see \Cref{fig:three-layered-partition}). Note that $w(P) \ge 2h_1(u) + 2$ after adding $1 + \max(2h_1(u),2h_3(u))$ to all edges incident to $u$ to make all edge weights positive in $G^*$, while the not-shortest path $P' := u \to z_2 \to v_e$ from $u$ to $v_e$ has weight $w(P') = w(P) + 1$. 

The issue with using this construction directly for the $\alpha$-stretch version is that when $h_1(u)$ is large (in our definition of $h_i$ in \cref{lem:h_i}, $h_1(u)$ gets exponentially large), we have $w(P') = w(P) + 1 < \alpha \cdot w(P)$ even for small $\alpha = 1 + \varepsilon$, so that although $P$ must remain shortest in exact shortest-paths preservers, the shortest path from $u$ to $v_e$ in an $\alpha$-stretch shortest-paths preserver $H$ could be either $P$ or $P'$. Without these shortest paths, we would lose the edge inequalities that gives exponential maximum weight in $H$.

To address this problem, our goal is to keep the same graph structure and re-define the edge weights so that the paths in $\mathcal P$ are not just shortest, but also the \emph{only} $\alpha$-approximate shortest path between their endpoints. This way, all paths in $\mathcal P$ must also be shortest in any $\alpha$-stretch gap-1 shortest-paths preserver $H$, so that $H$ has exponential maximum weight.

\subsection{Construction}

Fix $\alpha > 1$. As mentioned above, we re-use the same lower bound graph for the exact version, but with modified edge weights. Let $(G_1,G_2,G_3)$ where $G_i = (V,E_i)$ be the exponential lower bound additive triple from \cref{thm:three_dag_lb}, 
and let $G^* = (V^* = V_1^* \cup V_2^* \cup V_3^*, E^*, w)$ be the 3-layered DAG created via \cref{thm:reduction} 
using $(G_1,G_2,G_3)$. 

We define the lower bound graph $G = (V^*, E^*, w_\alpha)$ on the same vertex and edge sets as $G^*$, and assign new edge weights. We will assign some exponent $p(e)$ for each edge $e \in E^*$, and and set the new weights to be 
\[
w_\alpha(e) := (3\alpha)^{p(e)}.
\]
The intuition is that by choosing $p$ such that for each shortest path $P \in \mathcal P$ with $P = a \to b \to c$ and alternate path $P' = a \to b' \to c$ between the same endpoints (there is only one alternate path since each $c \in V_3^*$ has $\mathrm{indeg}(c) = 2$), the edges on $P'$ will have larger $p$ value in a way that gives the desired inequality $\alpha \cdot w_\alpha(P) < w_\alpha(P')$. The precise condition on what it means for edges on $P'$ to have larger $p$ value will be defined in \cref{subsec:an_approx}.

We define the exponents $p: E^* \to \mathbb Z_{\ge 0}$ below, where \Cref{fig:three-layered-partition-weighted} depicts a part of the new weights assignment, along with some paths relevant for the analysis.

\begin{itemize}
    \item 
    We first define $p$ for all edges between layers $V_1^*$ and $V_2^*$. Recall that the first layer $V_1^*$ contains one vertex for each vertex in $V$, which was defined as
    \[V = \{b_0,c_0\} \cup \bigcup_{r=1}^t \{a_r,b_r,c_r,d_r\}.
    \]
    For each $u \in V_1^*$, we assign a triple $p(u) := (p(u,z_1), p(u,z_2), p(u,z_3)) \in \mathbb Z_{\ge 0}^3$ for the exponent of all edges incident to $u$. 
    For the initial vertices we assign $p(b_0) := (0,0,0)$, $p(c_0) := (1,0,0)$. For each odd $r > 0$, we define 
    \[
    p(a_r) := 2\cdot (5r,5r+1,5r+3), \qquad p(b_r) := 2\cdot (5r+1, 5r, 5r+4),
    \]
    \[ p(c_r) := 2\cdot (5r, 5r+1, 5r), \qquad p(d_r) := 2\cdot (5r+1, 5r, 5r+2).
    \]
    For each even $r > 0$, we define
    \[
    p(a_r) := 2\cdot (5r+2, 5r, 5r+1), \qquad p(b_r) := 2\cdot (5r, 5r+1, 5r),
    \]
    \[
    p(c_r) := 2\cdot (5r+4, 5r, 5r+1), \qquad p(d_r) := 2\cdot (5r+3, 5r+1, 5r).
    \]
    The weight assignments might look quite arbitrary at first, but there is some symmetry to observe. We ignore the multiplicative factor of 2 for now since every vertex has that. In the $r$ odd case, $a_r$ and $c_r$ have first two indices $5r,5r+1$, and for  $b_r,d_r$ it's the same but in the opposite order. In the $r$ even case, the same holds but for the last two indices instead. The remaining indices (the third index for the $r$ odd case and the first index for the $r$ even case) have additive $0,2,3,4$ in some order. 
    
    \item 
    We now define $p$ for all edges between layers $V_2^*$ and $V_3^*$. Recall that $V_3^* = V_{3,1}^* \cup V_{3,2}^* \cup V_{3,3}^*$, where $V_{3,i}^* = \{v_e : e \in E_i\}$. 
    
    \begin{itemize}
        \item 
        For each $v_e \in V_{3,1}^*$ where $e = (u,v) \in E_1$, we define $p$ on the edges $(z_1,v_e), (z_2,v_e)$ incident to $v_e$ by
        \[
        (p(z_{1},v_e), p(z_{2},v_e)) := \begin{cases} (p(u, z_{2}) - 1, 0) & \text{if } p(u,z_{2}) \ge p(u,z_{1}) \\ (0, p(u,z_{1})-1) & \text{else} \end{cases}
        \]

        \item 
        For each $v_e \in V_{3,2}^*$ where $e = (u,v) \in E_2$, we define $p$ on the edges $(z_2,v_e), (z_3,v_e)$ incident to $v_e$ by
        \[
        (p(z_{2},v_e), p(z_{3},v_e)) := \begin{cases} (p(u, z_{3}) - 1, 0) & \text{if } p(u,z_{3}) \ge p(u,z_{2}) \\ (0, p(u,z_{2})-1) & \text{else} \end{cases}
        \]

        \item 
        For each $v_e \in V_{3,3}^*$ where $e = (u,v) \in E_3$, we define $p$ on the edges $(z_1,v_e), (z_3,v_e)$ incident to $v_e$ by
        \[
        (p(z_{1},v_e), p(z_{3},v_e)) := \begin{cases} (p(u, z_{3}) - 1, 0) & \text{if } p(u,z_{3}) \ge p(u,z_{1}) \\ (0, p(u,z_{1})-1) & \text{else} \end{cases}
        \]
    \end{itemize}

    This finishes the definition of $p : E^* \to \mathbb Z_{\ge 0}$.
    
\end{itemize}

Recall that for all $e \in E^*$, the new weight of $e$ is set to be
\[
w_\alpha(e) := (3\alpha)^{p(e)}.
\] 

\begin{figure}[htbp]
\centering

\begin{tikzpicture}[
vertex/.style={circle, draw, inner sep=1.2pt, minimum size=20pt},
blackedge/.style={->, >=Stealth, thick, draw=black},
rededge/.style={->, >=Stealth, thick, draw=red},
setbox/.style={draw, rounded corners, thick, inner sep=8pt},
layerlabel/.style={font=\large},
partlabel/.style={font=\normalsize},
pad/.style={inner sep=0pt, minimum size=0pt},
edgelabel/.style={fill=white, inner sep=1pt, font=\scriptsize}
]

\def\xone{-5}
\def\xtwo{0}
\def\xthree{5}

\node (V1topdots) at (\xone,1.8) {$\vdots$};
\node[vertex] (u) at (\xone,0.9) {$u$};
\node (V1middots) at (\xone,0.0) {$\vdots$};
\node[vertex] (v) at (\xone,-0.9) {$v$};
\node (V1botdots) at (\xone,-1.8) {$\vdots$};

\node[setbox, fit=(V1topdots)(u)(V1middots)(v)(V1botdots)] (V1box) {};
\node[layerlabel, above=4pt of V1box] {$V_1^* = V$};

\node[vertex] (z1) at (\xtwo,1.2) {$z_1$};
\node[vertex] (z2) at (\xtwo,0.0) {$z_2$};
\node[vertex] (z3) at (\xtwo,-1.2) {$z_3$};

\node[setbox, fit=(z1)(z2)(z3)] (V2box) {};
\node[layerlabel, above=4pt of V2box] {$V_2^*$};

\node[pad]    (V31top) at (\xthree,0.95) {};
\node[vertex] (ve)     at (\xthree,0.55) {$v_e$};
\node         (V3dots1) at (\xthree,-0.25) {$\vdots$};
\node[pad]    (V31bot) at (\xthree,-0.65) {};
\node[setbox, fit=(V31top)(ve)(V3dots1)(V31bot)] (V3box1) {};
\node[partlabel, right=8pt of V3box1] {$V_{3,1}^*$};

\node at (\xthree,-1.65) {$\vdots$};

\node[layerlabel] at (\xthree,1.8) {$V_3^*$};

\draw[rededge]
(u) -- node[edgelabel, pos=0.4, above=2pt] {$(3\alpha)^{p(u,z_1)}$} (z1);
\draw[blackedge]
(u) -- node[edgelabel, pos=0.4, above=0pt] {$(3\alpha)^{p(u,z_2)}$} (z2);

\draw[blackedge]
(v) -- node[edgelabel, pos=0.4,below=0pt] {$(3\alpha)^{p(v,z_1)}$} (z1);
\draw[rededge]
(v) -- node[edgelabel, pos=0.4, below=3pt] {$(3\alpha)^{p(v,z_2)}$} (z2);

\draw[rededge]
(z1) -- node[edgelabel, pos=0.5, above=3pt] {$(3\alpha)^{p(u,z_2)-1}$} (ve);
\draw[rededge]
(z2) -- node[edgelabel, pos=0.5, below=3pt] {$1$} (ve);

\end{tikzpicture}

\caption{The shortest paths in $\mathcal P$ (highlighted in red) incident to some $v_e \in V_{3,1}^*$ that are the only $\alpha$-approximate shortest path, under the labeled new edge weights $w_\alpha$. If $p(u,z_2) \ge p(v,z_1)$, the weight of edges incident to $v_e$ are labeled as above. Otherwise, we would have $w_\alpha(z_1,v_e) = 1$ and $w_\alpha(z_2,v_e) = (3\alpha)^{2p(v,z_1)-1}$. For $v_e \in V^*_{3,2}, V^*_{3,3}$, the weights are analogous. }
\label{fig:three-layered-partition-weighted}

\end{figure}
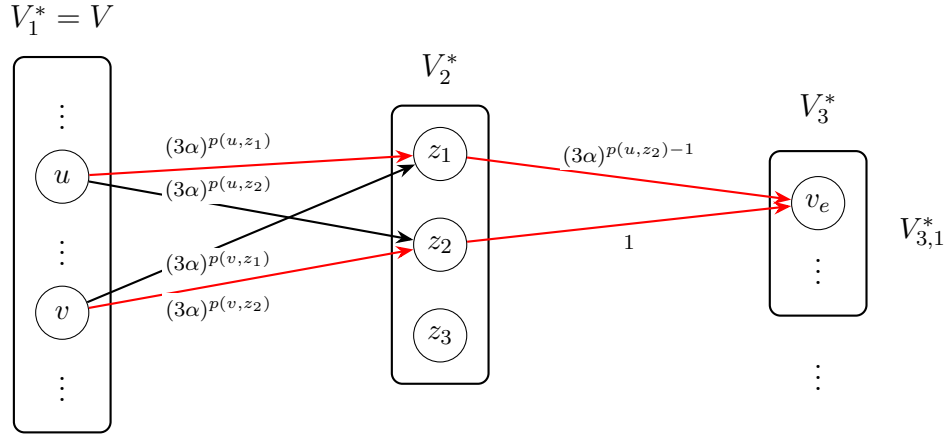

\subsection{Analysis}\label{subsec:an_approx}

\paragraph{Reminder of analysis from exact version.}
Recall that in the previous lower bound for the exact version, we proved in \cref{lem:differences} that for each $u \in V$, the edge weight $w$ satisfies
\[
w(u,z_1) - w(u,z_2) = 2h_1(u).
\]
Consequently, for each $e = (u,v) \in E_1$, we have $h_1(v) \ge h_1(u) + 1$ since $h_1$ is a topological ordering, so we have
\begin{equation}
\label{eq:2_slack}
    w(v,z_1) - w(v,z_2) \ge w(u,z_1) - w(v,z_2) + 2.
\end{equation}
By having a $+2$ slack, this allowed us to define edge weights incident to $v_e$ so that the difference $w(z_2, v_e) - w(z_1,v_e)$ satisfies
\[
w(u,z_1) - w(v,z_2) + 1 \le w(z_2, v_e) - w(z_1,v_e) \le w(v,z_1) - w(v,z_2) -1
\]
which shows that the shortest paths $u \to z_1 \to v_e, v \to z_2 \to v_e \in \mathcal P$ are indeed shortest by expanding the inequalities.

\paragraph{New analysis.}
Here, in the $\alpha$-stretch version, we adopt a similar strategy. We first prove that our weight assignment $w_\alpha$ satisfies an analogous condition to \cref{eq:2_slack}.

\begin{lemma}\label{lem:max}
    The following inequalities are true:
    \begin{itemize}
        \item 
        For each $(u,v) \in E_1$, we have
        \[
        \max\{ p(u,z_{2}), p(v,z_{1})\} \ge \max\{p(u,z_{1}), p(v,z_{2})\} + 2.
        \]

        \item 
        For each $(u,v) \in E_2$, we have
        \[
        \max\{ p(u,z_{3}), p(v,z_{2})\} \ge \max\{p(u,z_{2}), p(v,z_{3})\} + 2.
        \]

        \item 
        For each $(u,v) \in E_3$, we have
        \[
        \max\{ p(u,z_{3}), p(v,z_{1})\} \ge \max\{p(u,z_{1}), p(v,z_{3})\} + 2.
        \]
    \end{itemize}
\end{lemma}

\begin{proof}
    For each $i \in \{1,2,3\}$, we need to check the above for all edges in $E_i$ (where $E_i$ was defined in \cref{subsec:3dag_construction}). Note that all $p(u)$ are multiplied by a factor of 2, so we look at the triples $\frac 12p(u)$ and prove the same inequalities but with $+1$ instead of $+2$ on the right hand side.
    \begin{itemize}
        \item 
        $E_1$: For the initial edge $(b_0,c_0)$, we have 
        \[
        \max\{p(b_0,z_2), p(c_0,z_1)\} = 1 > 0 = \max\{p(b_0,z_1), p(c_0,z_2)\}.
        \]
        For each odd $r$, we have edges $\{(a_{r},b_{r-1}),(c_{r},b_{r-1}),(c_{r-1},b_r),(c_{r-1},d_r)\}$. We check that
        \[
        \max\{p(a_r,z_2), p(b_{r-1},z_1)\} = p(a_r,z_2) = 5r+1 > 5r = p(a_r,z_1) = \max\{p(a_r,z_1), p(b_{r-1},z_2)\} 
        \]
        \[
        \max\{p(c_r,z_2), p(b_{r-1},z_1)\} = p(c_r,z_2) = 5r+1 > 5r = p(c_r,z_1) = \max\{p(c_r,z_1), p(b_{r-1},z_2)\} 
        \]
        \[
        \max\{p(c_{r-1},z_2), p(b_r,z_1)\} = p(b_r,z_1) = 5r+1 > 5r = p(b_r,z_2) = \max\{p(c_{r-1},z_1), p(b_r,z_2)\} 
        \]
        \[
        \max\{p(c_{r-1},z_2), p(d_r,z_1)\} = p(d_r,z_1) = 5r+1 > 5r = p(d_r,z_2) = \max\{p(c_{r-1},z_1), p(d_r,z_2)\}.
        \]
        For each even $r$, we have the edge $(a_r,d_r)$. We check that
        \[
        \max\{p(a_r,z_2), p(d_r,z_1)\} = 5r+3 > 5r+2 = \max\{p(a_r,z_1), p(d_r,z_2)\}.
        \]

        \item 
        $E_2$: For each odd $r$, we have the edge $(a_r,d_r)$. We check that
        \[
        \max\{p(a_r,z_3), p(d_r,z_2)\} = 5r+3 > 5r+2 = \max\{p(a_r,z_2), p(d_r,z_3)\}.
        \]
        For each even $r$, we have edges $\{(a_{r},b_{r-1}),(c_{r},b_{r-1}),(c_{r-1},b_r),(c_{r-1},d_r)\}$. We check that
        \[
        \max\{p(a_r,z_3), p(b_{r-1},z_2)\} = p(a_r,z_3) = 5r+1 > 5r = p(a_r,z_2) = \max\{p(a_r,z_2), p(b_{r-1},z_3)\}
        \]
        \[
        \max\{p(c_r,z_3), p(b_{r-1},z_2)\} = p(c_r,z_3) = 5r+1 > 5r = p(c_r,z_2) = \max\{p(c_r,z_2), p(b_{r-1},z_3)\}
        \]
        \[
        \max\{p(c_{r-1},z_3), p(b_r,z_2)\} = p(b_r,z_2) = 5r+1 > 5r = p(b_r,z_3) = \max\{p(c_{r-1},z_2), p(b_r,z_3)\}
        \]
        \[
        \max\{p(c_{r-1},z_3), p(d_r,z_2)\} = p(d_r,z_2) = 5r+1 > 5r = p_3(d_r) = \max\{p(c_{r-1},z_2), p(d_r,z_3)\}.
        \]
        
        \item 
        $E_3$: For each $r$, we have the edges $\{(b_r, a_r), (d_r, c_r)\}$. For odd $r$, we check that
        \[
        \max\{p(b_r,z_3), p(a_r,z_1)\} = 5r+4 > 5r+3 = \max\{p(b_r,z_1), p(a_r,z_3)\}
        \]
        \[
        \max\{p(d_r,z_3), p(c_r,z_1)\} = 5r+2 > 5r+1 = \max\{p(d_r,z_1), p(c_r,z_3)\}.
        \]
        For even $r$, we check that
        \[
        \max\{p(b_r,z_3), p(a_r,z_1)\} = 5r+2 > 5r+1 = \max\{p(b_r,z_1), p(a_r,z_3)\}
        \]
        \[
        \max\{p(d_r,z_3), p(c_r,z_1)\} = 5r+4 > 5r+3 = \max\{p(d_r,z_1), p(c_r,z_3)\}.
        \]
        \end{itemize}
\end{proof}

Finally, we check that the shortest paths in $\mathcal P$ (the ones in \cref{lem:sp_structure}) are the only $\alpha$-approximate shortest paths between their endpoints in $G$. 

\begin{lemma}
\label{lem:unique}
    We have the following unique $\alpha$-stretch shortest paths under the new weighting $w_\alpha$:
    \begin{itemize}
        \item 
        For $e = (u,v) \in E_1$, the unique $\alpha$-stretch shortest path from $u$ to $v_e$ is $u \to z_{1} \to v_e$, and the unique $\alpha$-stretch shortest path from $v$ to $v_e$ is $v \to z_2 \to v_e$.

        \item 
        For $e = (u,v) \in E_2$, the unique $\alpha$-stretch shortest path from $u$ to $v_e$ is $u \to z_{2} \to v_e$, and the unique $\alpha$-stretch shortest path from $v$ to $v_e$ is $v \to z_3 \to v_e$.

        \item 
        For $e = (u,v) \in E_3$, the unique $\alpha$-stretch shortest path from $u$ to $v_e$ is $u \to z_{1} \to v_e$, and the unique $\alpha$-stretch shortest path from $v$ to $v_e$ is $v \to z_3 \to v_e$.
    \end{itemize}
\end{lemma}

Before the formal proof, we provide an outline for the case $(u,v) \in E_1$ and $p(u,z_2) \ge p(v,z_1)$ (other cases are similar). See \Cref{fig:three-layered-partition-weighted} for the relevant paths. We look at the paths between $u$ and $v_e$. By \cref{lem:max}, we have that $p(u,z_2) \ge p(u,z_1) + 2$. Since $p(z_1,v_e) = p(u,z_2) - 1$, we know that $p(u,z_2)$ is at least 1 bigger than both $p(u,z_1)$ and $p(z_1,v_e)$. Intuitively, this is good because an edge on the not-shortest path has the largest exponent. Then, by the choice of the base $3\alpha$, we would have the following inequalities (details are verified in the proof of \cref{lem:unique} below):
\[
w_\alpha(u \to z_2 \to v_e) > (3\alpha)^{p(u,z_2)} > \alpha\cdot ((3\alpha)^{p(u,z_1)} + (3\alpha)^{p(z_1,v_e)}) = \alpha \cdot w_\alpha(u \to z_1 \to v_e)
\]
so the claim is true for $u \to z_1 \to v_e \in \mathcal P$. Now we look at the paths between $v$ and $v_e$. Again by \cref{lem:max}, we have that $p(z_1,v_e) = p(u,z_2)-1 \ge p(v,z_2)+1$. Then we would have the following inequalities (again, details are verified in the proof of \cref{lem:unique} below):
\[
w_\alpha(v \to z_1 \to v_e) > (3\alpha)^{p(z_1,v_e)} > \alpha\cdot ((3\alpha)^{p(v,z_2)} + 1) = \alpha \cdot w_\alpha(v \to z_2 \to v_e)
\]
proving that the claim is true for $v \to z_2 \to v_e \in \mathcal P$ as well.

\begin{proof}[Proof of \cref{lem:unique}]
    We look at the case $e = (u,v) \in E_1$. It is enough to show that 
    \[
    \alpha \cdot w_\alpha(u \to z_1 \to v_e) < w_\alpha(u \to z_2 \to v_e)
    \]
    and
    \[
    \alpha \cdot w_\alpha(v \to z_2 \to v_e) < w_\alpha(v \to z_1 \to v_e).
    \]

    We do casework on whether $p(u,z_{2}) \ge p(v,z_{1})$: 
    \begin{itemize}
        \item 
        Suppose $p(u,z_{2}) \ge p(v,z_{1})$. So \cref{lem:max} tells us that
        \[
        p(u,z_{2}) \ge  \max\{p(u,z_{1}), p(v,z_{2})\} + 1.
        \]
        Then we have
        \begin{align*}
            \alpha \cdot w_\alpha(u \to z_{1} \to v_e) &= \alpha \cdot \left((3\alpha)^{2p(u,z_{1})} + (3\alpha)^{2p(u,z_{2}) - 1}\right) \\
            &\le \alpha \cdot \left(2 \cdot (3\alpha)^{2p(u,z_{2}) -1}\right) \\
            &< (3\alpha)^{2p(u,z_{2})} \\
            &= w_\alpha(u, z_{2}) \\
            &< w_\alpha(u \to z_{2} \to v_e)
        \end{align*}
        and
        \begin{align*}
            \alpha \cdot w_\alpha(v \to z_{2} \to v_e) &= \alpha \cdot \left((3\alpha)^{2p(v,z_{2})} + 1\right) \\
            &\le \alpha \cdot \left(2 \cdot (3\alpha)^{2p(v,z_{2})}\right) \\
            &< (3\alpha)^{2p(v,z_{2}) + 1} \\
            &\le (3\alpha)^{2p(u,z_{2})-1} \\
            &= w_\alpha(z_{1}, v_e) \\
            &< w_\alpha(v \to z_{1} \to v_e).
        \end{align*}

        \item 
        Suppose $p(u,z_{2}) < p(v,z_{1})$, so we have
        \[
        p(v,z_{1}) \ge \max\{p(u,z_{1}), p(v,z_{2})\} + 1.
        \]
        The argument is almost symmetric, in the sense that the roles of $(z_1,v_e)$ and $(z_2,v_e)$ are flipped, so the inequalities for $u \to z_1 \to v_e$ and $v \to z_2 \to v_e$ are flipped. 
        We have
        \begin{align*}
            \alpha \cdot w_\alpha(u \to z_{1} \to v_e) &= \alpha \cdot \left((3\alpha)^{2p(u,z_{1})} + 1\right) \\
            &\le \alpha \cdot \left(2 \cdot (3\alpha)^{2p(u,z_{1})}\right) \\
            &< (3\alpha)^{2p(u,z_{1})+1} \\
            &\le (3\alpha)^{2p(v,z_{1})-1} \\
            &= w_\alpha(z_{2}, v_e) \\
            &< w_\alpha(u \to z_{2} \to v_e)
        \end{align*}
        and
        \begin{align*}
            \alpha \cdot w_\alpha(v \to z_{2} \to v_e) &= \alpha \cdot \left((3\alpha)^{2p(v,z_{2})} + (3\alpha)^{2p(v,z_{1})-1}\right) \\
            &\le \alpha \cdot \left(2 \cdot (3\alpha)^{2p(v,z_{1})-1}\right) \\
            &< (3\alpha)^{2p(v,z_{1})} \\
            &= w_\alpha(v,z_{1}) \\
            &< w_\alpha(v \to z_{1} \to v_e).
        \end{align*}
    \end{itemize}

    This proves the case $(u,v) \in E_1$. For the case $e = (u,v) \in E_2$, we can repeat the proof above with $z_1,z_2$ replaced by $z_2,z_3$, respectively. Everything works the same since the definition of $p$ on edges $(z_2,v_e),(z_3,v_e)$ and \cref{lem:max} for $e \in E_2$ is the same as $e \in E_1$ except with $z_1,z_2$ replaced by $z_2,z_3$. For the case $e = (u,v) \in E_3$, we can repeat the proof above with $z_2$ replaced by $z_3$. 
\end{proof}

This finishes the proof of \cref{thm:approx}.

\section{Linear Upper Bound for $|V_2| = 2$ Case}\label{sec:ub}

We prove \cref{thm:upper_bound}: 

\upperbound*

Our strategy to prove \Cref{thm:upper_bound} is the follows: Given any 3-layered DAG $G$, we first construct an auxiliary DAG $G^*$. Using the topological order of $G^*$, we define edge weights in $H$ so that all shortest paths are preserved.

Note that without loss of generality,  shortest paths are unique (that is, for each pair of endpoints $(s,t)$ such that $s$ can reach $t$, there is a unique shortest path between them). This is because if shortest paths are not unique, we can make them so using the standard trick of applying a random perturbation to the edge weights. This is sound because in the case of ties shortest-paths preservers only need to preserve one shortest path for each pair of endpoints.

Let $G = (V_1\cup V_2 \cup V_3, E, w)$ be any 3-layered DAG with $|V_2| = 2$. Let $V_2 = \{z_1,z_2\}$. The reduction goes as follows: We construct a graph $G^* = (V^*, E^*)$ such that
\begin{itemize}
    \item 
    $V^* = V_1 \cup V_3$

    \item 
    For each pair $(u,v) \in V_1 \times V_3$, if the shortest path from $u$ to $v$ is $u \to z_1 \to v$, add the directed edge $(u,v)$ to $E^*$. Otherwise, if the shortest path from $u$ to $v$ is $u \to z_2 \to v$, add the directed edge $(v,u)$ to $E^*$.
\end{itemize}

\begin{lemma}
    $G^*$ is acyclic.
\end{lemma}

\begin{proof}
    Suppose for contradiction that $G^*$ has a cycle. Note that by definition $G^*$ is bipartite with edges between $V_1$ and $V_3$, so there exists $u_1,\ldots,u_k \in V_1$ and $v_1,\ldots,v_k \in V_3$ such that $(u_1,v_1,u_2,v_2,\ldots,u_k,v_k,u_1)$ forms a cycle. We will show that the edges in this cycle in $G^*$ imply there is a set of shortest paths in $G$ that are incompatible with each other, which gives a contradiction. 

    For each edge $(u_i,v_i)$ where $i \in [k]$, we know that the shortest path from $u_i$ to $v_i$ in $G$ is $u_i \to z_1 \to v_i$, so
    \[
    w(u_i,z_1) + w(z_1,v_i) < w(u_i,z_2) + w(z_2,v_i).
    \]
    For each edge $(v_i,u_{i+1})$ where $i \in [k]$ and $u_{k+1} := u_1$, we know that the shortest path from $u_{i+1}$ to $v_i$ in $G$ is $u_{i+1} \to z_2 \to v_i$, so
    \[
    w(u_{i+1},z_2) + w(z_2,v_i) < w(u_{i+1},z_1) + w(z_1,v_i).
    \]
    Summing up the inequalities above over all $i \in [k]$, we have
    \[
    \sum_{i=1}^k w(u_i,z_1) + w(z_1,v_i) + w(u_{i+1},z_2) + w(z_2,v_i) < \sum_{i=1}^k w(u_i,z_2) + w(z_2,v_i) + w(u_{i+1},z_1) + w(z_1,v_i).
    \]
    Note that since both summations are over all $i$, we can change the $u_{i+1}$ on both sides to $u_i$ without changing the sum, after which we can see that both sides are equal. This gives a contradiction.  Thus, $G^*$ must be acyclic.
\end{proof}

Since $G^*$ is acyclic, it has some topological ordering $h: V^* \to \{1,2,\ldots,|V^*|\}$. Using this, we are ready to prove the original theorem:

\begin{proof}[Proof of ~\cref{thm:upper_bound}]

    To construct the shortest-paths preserver $H$, we define $w_H$ as follows: 
    \begin{itemize}
        \item 
        For each $u \in V_1$, let $w_H(u,z_1) := 1 + h(u)$ and $w_H(u,z_2) := 1$.

        \item 
        For each $v \in V_3$, let $w_H(z_1,v) := 1$ and $w_H(z_2,v) := 1 + h(v)$.
    \end{itemize}

    It is clear that all edge weights are positive integers that are at most $1 + |V^*| = 1 + |V_1| + |V_3|$, so it remains to check that $w_H$ preserves all the shortest paths.
    
    For each $(u,v) \in V_1 \times V_3$:
    \begin{itemize}
        \item 
        If the shortest path from $u$ to $v$ in $G$ is $u \to z_1 \to v$, we have $(u,v) \in E^*$, so $h(u) < h(v)$. We have
        \[
        w_H(u,z_1) + w_H(z_1,v) = 2 + h(u) < 2 + h(v) = w_H(u,z_2) + w_H(z_2,v),
        \]
        so $u \to z_1 \to v$ is also shortest in $H$.

        \item 
        Similarly, if the shortest path from $u$ to $v$ in $G$ is $u \to z_2 \to v$, we have $(v,u) \in E^*$, so $h(v) < h(u)$. We have
        \[
        w_H(u,z_2) + w_H(z_2,v) = 2 + h(v) < 2 + h(u) = w_H(u,z_1) + w_H(z_1,v),
        \]
        so $u \to z_2 \to v$ is also shortest in $H$.
    \end{itemize}

    So $H$ preserves all shortest paths in $G$.
\end{proof}

\subsection*{Acknowledgments}

We thank Greg Bodwin for suggesting looking at the 3-layered case.
\label{sec:acks}

\paragraph{Use of AI Tools.} We used ChatGPT 5.5 Pro to assist with finding the lower bound constructions in \cref{sec:3DAG} and \cref{sec:approx}. The tool materially affected \cref{sec:3DAG} and \cref{sec:approx}. We did not use AI for the proofs in \cref{sec:reduction}, \cref{sec:ub}, or \cref{sec:other}. More details can be found in \cref{subsec:techniques}. For all content that AI assisted with, the authors verified its correctness and heavily rewrote it.

\bibliographystyle{alpha}
\bibliography{references}

\appendix
\section{Reduction in the Other Direction}\label{sec:other}

We prove~\Cref{thm:other_direction}:

\otherdir*

\subsection{Construction}

Let $V_2 = \{z_1, z_2, z_3\}$. We assume the input graph $G$ is a complete 3-layered DAG, that is, there is an edge $(u,z_i)$ for all $u \in V_1, z_i \in V_2$ and there is an edge $(z_i,v)$ for all $z_i \in V_2, v \in V_3$. 

Note that this assumption is without loss of generality because the case where the input instance is complete is the ``hardest'': if $G$ is not complete, we add all the missing edges and set their weight to $W$ for some sufficiently large $W$ so that none of the added edges lie on shortest paths to obtain $G_{\text{complete}}$. Then for any shortest-paths preserving reweighting of $G_{\text{complete}}$, we obtain a shortest-paths preserving reweighing of $G$ simply by ignoring the weight of the added edges. Then any upper bound on the optimal value (smallest maximum edge weight of a shortest-paths reweighting) of $G_{\text{complete}}$ gives an upper bound for the optimal of $G$, so it is enough to consider $G_{\text{complete}}$. 

We construct $(G_1^*, G_2^*, G_3^*)$ where $G_i^* = (V^*, E_i^*)$ as follows: 

\begin{itemize}
    \item 
    Let $V^* = V_1 \cup V_3$.

    \item 
    For each pair of endpoints $u \in V_1, v \in V_3$:
    \begin{itemize}
        \item 
        If $u \to z_1 \to v$ is the shortest path in $G$, add the edge $(u,v)$ to $E_1^*$ and $E_3^*$.

        \item 
        If $u \to z_2 \to v$ is the shortest path in $G$, add the edge $(v,u)$ to $E_1^*$, and add the edge $(u,v)$ to $E_2^*$.

        \item 
        If $u \to z_3 \to v$ is the shortest path in $G$, add the edge $(v,u)$ to both $E_2^*$ and $E_3^*$.
    \end{itemize}
\end{itemize}

\subsection{Analysis}

It is clear that $|V^*| = |V_1| + |V_3|$ and $|E_1^*| + |E_2^*| + |E_3^*| = 2 \cdot |V_1| \cdot |V_3|$. 

We first show that $(G_1^*, G_2^*, G_3^*)$ is an additive triple of DAGs. Right now, it is not even clear that they are DAGs.

\begin{lemma}
    For each $u \in V^*$, if $u \in V_1$, let
    \[
    h_1^*(u) := w(u,z_1) - w(u,z_2), \quad h_2^*(u) := w(u,z_2) - w(u,z_3), \quad h_3^*(u) = w(u,z_1) - w(u,z_3).
    \]
    If $u \in V_3$, let
    \[
    h_1^*(u) := w(z_2,u) - w(z_1,u), \quad h_2^*(u) := w(z_3,u) - w(z_2,u), \quad h_3^*(u) := w(z_3,u) - w(z_1,u).
    \]
    Then, each $h_i^*$ is a valid topological ordering for $G_i^*$, and $h_1^* + h_2^* = h_3^*$.
\end{lemma}

\begin{proof}
    It is clear by definition that $h_1^* + h_2^* = h_3^*$ since $-w(u,z_2)$ and $w(u,z_2)$ cancels out. We will check that $h_1^*$ is a topological ordering of $G_1^*$.

    For the edges $(u,v) \in E_1^*$ where $u \to z_1 \to v$ is the shortest path in $G$, by the shortest path condition, we have
    \[
    w(u,z_1) + w(z_1,v) < w(u,z_2) + w(z_2,v) \implies w(u,z_1) - w(u,z_2) < w(z_2,v) - w(z_1,v),
    \]
    so $h_1^*(u) < h_1^*(v)$. Similarly, for the edges $(v,u) \in E_1^*$ where $u \to z_2 \to v$ is the shortest path in $G$, we have
    \[
    w(u,z_2) + w(z_2,v) < w(u,z_1) + w(z_1,v) \implies w(z_2,v) - w(z_1,v) < w(u,z_1) - w(u,z_2),
    \]
    so $h_1^*(v) < h_1^*(u)$. Therefore, $h_1^*$ is a valid topological ordering of $G_1^*$. 

    The proof is the same for $h_2^*$ with $z_1,z_2$ replaced by $z_2,z_3$ respectively, and the proof is the same for $h_3^*$ with $z_2$ replaced by $z_3$. 

    Note that for all edges, say $(u,v) \in E_1^*$, we only proved $h_1^*(u) < h_1^*(v)$ instead of $h_1^*(u) + 1 \le h_1^*(v)$ which is required in the definition. This is okay because we can multiply all $h_i^*$ by some large constant $C' > 0$ to boost the difference.
\end{proof}

Let $h_i^*$ be any topological ordering of $G_i^*$ for each $i$, with $h_1^* + h_2^* = h_3^*$. We now show that we can get a shortest-paths preserver $G''$ using the topological orderings. 

\begin{lemma}
    Let $G'' = (V, E, w'')$ where $w''$ is defined as follows:
    \begin{itemize}
        \item 
        For each $u \in V_1$, let
        \[
        w''(u,z_1) := h_3^*(u), \qquad w''(u,z_2) := h_2^*(u), \qquad w''(u,z_3) := 0.
        \]

        \item 
        For each $v \in V_3$, let
        \[
        w''(z_1,v) := -h_3^*(v), \qquad w''(z_2,v) := -h_2^*(v), \qquad w''(z_3,v) := 0.
        \]
    \end{itemize}

    Then, $G''$ is a gap-1 shortest-paths preserver of $G$.
\end{lemma}

\begin{proof}
    Fix some pair $u \in V_1, v \in V_3$. The idea is that by the construction of $G_i^*$ and $w''$, choosing valid topological orders translates directly to choosing valid shortest-paths reweighting of $G$. 
    
    We do casework on the shortest path from $u$ to $v$ in $G$, and show that in each case the shortest paths are preserved:
    \begin{itemize}
        \item 
        Suppose $u \to z_1 \to v$ is the shortest path in $G$. Then we have $(u,v) \in E_1^*$ and $(u,v) \in E_3^*$, so $h_1^*(v)\ge 1+ h_1^*(u)$ and $h_3^*(v) \ge1+ h_3^*(u)$. We check that
        \[
        w''(u \to z_1 \to v) - w''(u \to z_2 \to v) = (h_3^* - h_2^*)(u) + (h_2^* - h_3^*)(v) = h_1^*(u) - h_1^*(v) \le-1
        \]
        and
        \[
        w''(u \to z_1 \to v) - w''(u \to z_3 \to v) = h_3^*(u) - h_3^*(v) \le-1,
        \]
        so $u \to z_1 \to v$ is still shortest in $G''$ with gap 1 to all not-shortest paths.

        \item 
        Suppose $u \to z_2 \to v$ is the shortest path in $G$. Then we have $(v,u) \in E_1^*$ and $(u,v) \in E_2^*$, so $h_1^*(u) \ge1+ h_1^*(v)$ and $h_2^*(v) \ge1+ h_2^*(u)$. We check that
        \[
        w''(u \to z_2 \to v) - w''(u \to z_1 \to v) = (h_2^* - h_3^*)(u) + (h_3^* - h_2^*)(v) = h_1^*(v) - h_1^*(u) \le-1
        \]
        and
        \[
        w''(u \to z_2 \to v) - w''(u \to z_3 \to v) = h_2^*(u) - h_2^*(v) \le-1,
        \]
        so $u \to z_2 \to v$ is still shortest in $G''$ with gap 1 to all not-shortest paths.

        \item 
        Suppose $u \to z_3 \to v$ is the shortest path in $G$. Then we have $(v,u) \in E_2^*$ and $(v,u) \in E_3^*$, so $h_2^*(u) \ge1+ h_2^*(v)$ and $h_3^*(u) \ge1+h_3^*(v)$. We check that
        \[
        w''(u \to z_3 \to v) - w''(u \to z_1 \to v) = h_3^*(v) - h_3^*(u) \le-1
        \]
        and
        \[
        w''(u \to z_3 \to v) - w''(u \to z_2 \to v) = h_2^*(v) - h_2^*(u) \le-1
        \]
        so $u \to z_3 \to v$ is still shortest in $G''$.
    \end{itemize}
    So $G''$ is a gap-1 shortest-paths preserver of $G$.
\end{proof}

Now, we show that we can modify the above gap-1 shortest-paths preserving reweighting to also have all edge weights positive and bounded by $\mathrm{range}(h_1^*, h_2^*, h_3^*)$.

\begin{lemma}
\label{lem:G'}
    Let $G' = (V, E, w')$ and 
    $T := \mathrm{range}(h_1^*, h_2^*, h_3^*)$.
    We define $w'$ as follows:
    \begin{itemize}
        \item 
        For each $u \in V_1$, let
        \[
        w'(u,z_1) := 1 + h_3^*(u) - \min_{u' \in V^*} h_3^*(u'), \quad w'(u,z_2) := 1 + h_2^*(u) - \min_{u' \in V^*} h_2^*(u'),
        \quad w'(u,z_3) := 1.
        \]

        \item
        For each $v \in V_3$, let
        \[
        w'(z_1,v) := 1 + T - h_3^*(v) + \min_{u' \in V^*} h_3^*(u'), \quad w'(z_2,v) := 1 + T - h_2^*(v) + \min_{u' \in V^*} h_2^*(u'),\quad\]
        \[
        w'(z_3,v) := 1 + T.
        \]
    \end{itemize}
    Then, $G'$ is a positively-weighted gap-1 shortest-paths preserver of $G$, and 
    \[
    1 \le w'(e) \le 1+T
    \]
    for all $e \in E$.
\end{lemma}

\begin{proof}
    We first prove that $w'$ is a gap-1 shortest-paths preserver of $G''$, which would imply it is a gap-1 shortest-paths preserver of $G$. In fact, we will prove something stronger: for each each pair of endpoints $u \in V_1, v \in V_3$, $w'$ preserves the relative difference between all paths in $G''$.

    Fix some pair of endpoints $u \in V_1, v \in V_3$. Check that
    \begin{align*}
    w'(u \to z_1 \to v) - w'(u \to z_3 \to v) &= (2+T + h_3^*(u) - h_3^*(v)) - (2+T) \\
    &= h_3^*(u) - h_3^*(v) \\
    &= w''(u \to z_1 \to v) - w''(u \to z_3 \to v) 
    \end{align*}
    and similarly
    \begin{align*}
    w'(u \to z_2 \to v) - w'(u \to z_3 \to v) &= (2+T + h_2^*(u) - h_2^*(v)) - (2+T) \\
    &= h_2^*(u) - h_2^*(v) \\
    &= w''(u \to z_2 \to v) - w''(u \to z_3 \to v).
    \end{align*}
    This proves that all relative differences of paths between $u,v$ are preserved, so the shortest paths are preserved and the gap to not-shortest paths remains at least 1.
    
    Finally, note that for all $u \in V^*$ and $i \in \{2,3\}$, we have
    \[
    0 \le h_i^*(u) - \min_{u' \in V^*} h_i^*(u') \le T,
    \]
    so the bound on the edge weights clearly hold.
\end{proof}

Recall that the original inequality we are trying to prove is
\begin{equation}
\label{ineq:other_direction}
\min_{G' = (V, E, w')} \left(\max_{e \in E} w'(e)\right) \le 1 + \min_{(h_1^*, h_2^*, h_3^*)} \left( \mathrm{range}(h_1^*, h_2^*, h_3^*)\right)
\end{equation}
where the first minimum is over gap-1 shortest-paths preservers $G' = (V, E, w')$ of $G$ and the second minimum is over 3-topological order preservers $(h_1^*, h_2^*, h_3^*)$ of $(G_1^*, G_2^*,G_3^*)$.

In \cref{lem:G'} we showed that for any 3-topological orderings $(h_1^*,h_2^*,h_3^*)$, there exists a gap-1 shortest-paths preserver $G' = (V,E,w')$ with
\[
\max_{e \in E} w'(e) \le 1 + \mathrm{range}(h_1^*, h_2^*,h_3^*).
\]
In particular, this inequality holds for the 3-topological ordering $(h_1^*,h_2^*,h_3^*)$ with minimum possible range. Since the left hand side of \cref{ineq:other_direction} takes a minimum over all gap-1 shortest-paths preservers and $G'$ is a gap-1 shortest-paths preserver by \cref{lem:G'}, we have proved \cref{ineq:other_direction}.

This completes the proof of \cref{thm:other_direction}.

\end{document}